\documentclass[10pt,onecolumn]{IEEEtran}
\usepackage[latin9]{inputenc}
\usepackage{amsmath,nccmath}
\usepackage{amssymb}
\usepackage{graphicx}
\usepackage{esint}
\usepackage[section]{placeins}
\usepackage{amsmath,amssymb,amsthm}
\usepackage{thmtools,thm-restate}
\usepackage{hyperref}
\usepackage{cleveref}
\usepackage{ifpdf}
\usepackage{cite}
\ifCLASSINFOpdf
\else
\fi
\usepackage{color}
\usepackage{placeins}
\usepackage{float}
\usepackage{tabularx}
\usepackage{colortbl}\usepackage{amsthm}

\usepackage{thmtools}
\begin{document}
\pdfoutput=1
\title{Search for Smart Evaders with Swarms of Sweeping Agents}

\author{Roee M. Francos$^{1}$ and Alfred M. Bruckstein$^{1}$
\thanks{$^{1}$Roee M. Francos and Alfred M. Bruckstein are with the Faculty
of Computer Science, Technion- Israel Institute of Technology, Haifa, Israel, 32000.
        {\tt\small roee.francos@cs.technion.ac.il},
        {\tt\small alfred.bruckstein@cs.technion.ac.il}}%
}

\maketitle

\begin{abstract}
Suppose that in a given planar circular region, there are some smart mobile evaders and we would like to find them using sweeping agents. We assume that each agent has a line sensor of length $2r$. We propose procedures for designing cooperative sweeping processes that ensure the successful completion of the task, thereby deriving conditions on the sweeping velocity of the agents and their paths. Successful completion of the task means that evaders with a given limit on their velocity cannot escape the sweeping agents. A simpler task for the sweeping swarm is the confinement of the evaders to their initial domain. The feasibility of completing these tasks depends on geometric and dynamic constraints that impose a lower bound on the velocity that the sweeper swarm must have. This critical velocity is derived to ensure the satisfaction of the confinement task. Increasing the velocity above the lower bound enables the agents to complete the search task as well. We present results on the total search time as a function of the sweeping velocity of the swarm's agents given the initial conditions on the size of the search region and the maximal velocity of the evaders.
\end{abstract}


\section{Introduction}
An interesting challenge for multi-agent systems is the design of searching or sweeping algorithms for static or mobile targets in a region, which can either be fully mapped in advance or unknown. Often the aim is to continuously patrol a domain in order to detect intruders or to systematically search for mobile targets known to be located within some area. Search for static targets involves complete covering of the area where they are located, but a much more interesting and realistic scenario is the question of how to efficiently search for targets that are dynamic and smart. A smart target is one that detects and responds to the motions of searchers by performing optimal evasive maneuvers, to avoid interception.

Several such problems originated in the second world war due to the need to design patrol strategies for aircraft aiming to detect ships or submarines in the English channel, see e.g. \cite{koopman1980search}. The problem of patrolling a corridor using multi agent sweeping systems in order to ensure the detection and interception of smart targets was also investigated by Vincent et. al. in \cite{vincent2004framework} and provably optimal strategies were provided by Altshuler et.al. in \cite{altshuler2008efficient}. A somewhat related, discrete version of the problem, was also investigated by Altshuler et. al. in \cite{altshuler2011multi}. It focuses on a dynamic variant of the cooperative cleaners problem, a problem that requires several simple agents to a clean a connected region on the grid with contaminated pixels. This contamination is assumed to spread to neighbors at a given rate. In \cite{mcgee2006guaranteed}, McGee et. al. investigate a search problem for smart targets. The targets do not have any maneuverability restrictions except for the maximal velocity they can move in and the sensor that the agents are equipped with detects all targets within a disk shaped area around the searcher location. The work of McGee et. al. \cite{mcgee2006guaranteed} consider search patterns consisting of spiral and linear sections. In \cite{hew2015linear}, Hew et. al. consider searching for smart evaders using concentric arc trajectories with agents sensors similar to \cite{mcgee2006guaranteed}. Such a search is proposed for detecting submarines in a channel or in an half plane. The paper focuses on determining the size of a region that can be successfully patrolled by a single searcher, where the searcher and evader velocities are known. The search problem in the paper is formulated as an optimization problem so that the search progress per arc or linear iteration, has to be maximized while guaranteeing that the evader cannot slip past the searcher undetected.

In our previous work \cite{francos2019search}, the confinement and cleaning tasks for a line formation of agents or alternatively for a single agent with a linear sensor are analyzed in terms of task completion times, i.e. the time at which all potential targets that resided in the initial evader region were detected. To guarantee the detection of all targets, the search is terminated when the area that contains potential evaders reaches $0$. A global consideration of swept areas versus the expanding area of the "danger zone" or the "evaders' possible locations" yields a lower bound on a searcher velocity that is independent of the search process. Several methods are proposed on how to determine the minimal velocity a single sweeper agent should have, in order to shrink the evader region to be bounded by a circle with a smaller radius than the searcher's sensor length. In \cite{francos2019search} an additional requirement on the ratio between the searcher velocity and the maximal evader expansion velocity that must be satisfied in order for a single agent or a line formation of agents to completely clean the evader region is provided. In \cite{francos2019search} for the case the agent or the linear formation of agents travels in a circular pattern around the evader region, it is proved that the minimal agent velocity has to be more than twice the lower bound. In \cite{francos2019search} it is proven that a single agent that employs a circular search around the evader region cannot completely clean the evader region without modifying its search pattern. In \cite{francos2019search} it is also proposed that in order to completely clean the evader region a modification for the sweep process which the single agent employs after the evader region is bounded by a circle with a radius of less than $r$ must be made.

Our work considers a scenario in which a multi-agent swarm of identical agents search for mobile targets or evaders that are to be detected. The information the agents perceive only comes from their own sensors, and evaders that intersect a sweeper's field of view are detected. We assume that all agents have a linear sensor of length $2r$. The evaders are initially located in a disk shaped region of radius $R_0$. There can be many evaders we wish do detect, and we consider the domain to be continuous, meaning that an evader can be located at any point in the interior of the circular region at the beginning of the search process. The sweepers are designed in a way that will require a minimal amount of memory in order to complete the required task due to the fact that the sweeping protocol is predetermined and deterministic. All sweepers move with a speed of $V_s$ (measured at the center of the linear sensor). By assumption the evaders move at a maximal speed of $V_T$, without any maneuverability restrictions. The sweeper swarm's objective is to "clean" or to detect all evaders that can move freely in all directions from their initial locations in the circular region of radius $R_0$. The search time will clearly depend on the type of sweeping movement the swarm employs.
The detection of evaders is done using deterministic and preprogrammed search protocols around the region.
We consider two types of search patterns, circular and spiral patterns. The desired result is that after each sweep around the region, the radius of the circle that bounds the evader region for the circular sweep, or the actual radius of the evader region for the spiral sweep, will decrease by a value that is strictly positive. This will guarantee a complete cleaning of the evader region, by shrinking the possible area in which evaders can reside to zero, in finite time. At the beginning of the circular search process we assume that only half the length of the agents sensors is inside the evader region, i.e. a footprint of length $r$, while the other half is outside the region in order to catch evaders that may move outside the region while the search progresses. At the beginning of the spiral search process we assume that the entire length of the agents sensors is inside the evader region, i.e. a footprint of length $2r$. We analyze the proposed sweep processes' performance in terms of the total time to complete the search, defined as the time at which all potential evaders that resided in the initial evader region were found. At first we provide a global balance of covered areas argument that is derived from a maximal swept area versus a minimal danger zone expansion area. This argument results in an equation that yields a lower bound on a searcher velocity that is independent of the search process. Secondly, we examine the performance of a multi agent swarm that performs a proposed circular sweep process. A critical velocity that depends on this circular search process is then derived. The proposed circular search pattern ensures the satisfaction of the confinement task. The developed circular critical velocity is compared to the lower bound on the critical velocity. We then show that the minimal agent velocity that ensures the satisfaction of the confinement task for the proposed circular search process is equal to twice the lower bound and hence is not optimal. Expressions for the complete cleaning times of the evader region as a function of the search parameters ,$R_0,r,V_T$ and the number of agents in the swarm are derived, evaluated and discussed. While the purpose of designing a circular search process is to perform the task with simple agents, it is not optimal. Therefore, the presented circular search pattern is improved into a proposed novel multi agent swarm spiral sweep process that uses spiral scans that draw inspiration from a previous work of McGee, \cite{mcgee2006guaranteed}. The proposed pattern tracks the wavefront of the expanding evader region and strives to have as close as possible to optimal sensor footprint over the evader region. Based on this proposed search pattern we develop an additional critical velocity that ensures the satisfaction of the confinement task for multi agent swarms that employ the spiral sweep process. We show that the developed spiral critical velocity approaches the theoretical optimal critical velocity that is independent of the search process.  Expressions for the complete cleaning times of the evader region as a function of the search parameters and the number of agents in the swarm are derived, evaluated and discussed for the spiral search process as well. Finally, we compare the different search methods, for both the circular and spiral sweep processes, in terms of completion times of the sweep processes. When comparing the different search processes we compare both the total cleaning times as well as the minimal searcher velocity that are required for a successful search.

As opposed to our work \cite{mcgee2006guaranteed} uses a disk shaped sensor with a radius of $r$, and does not calculate the time it takes to find all evaders. Furthermore, in our proposed sweep processes the initial positions of the sweeper agents are different from the initial placements of agents in \cite{mcgee2006guaranteed}.
In \cite{hew2015linear} the searcher uses a circular sensor of radius $r$ that detects evaders if and only if they are at a distance of at most $r$ from the searcher differing from the linear sensors used in our work.

This report is organized as follows. Section $\text{II}$ proves an optimal bound on the cleaning rate for a swarm that is independent of the search process that is deployed. This bound will serve as one of the benchmarks for comparing the performance of different search algorithms. In section $\text{III}$ the results for the completion of the search process for a swarm of sweeping agents that employ the circular sweep process are presented. In section $\text{IV}$ we perform an analysis for the case where the swarm employs the spiral sweep process. In section $\text{V}$ we provide a comparative unified analysis of the proposed search strategies that were developed in the previous sections. In the last section conclusions are given and future research directions are discussed.

\section{A Universal Bound On Cleaning Rate}
\noindent In this section we present an optimal bound on the cleaning rate of a searcher with a linear shaped sensor. This bound is independent of the particular search pattern employed. For each of the proposed search methods we then compare the resulting cleaning rate to the optimal derived bound in order to compare between different search methods. We will denote the searcher's velocity as $V_s$, the sensor length as $2r$ and the maximal velocity of an evading agent as $V_T$. The maximal cleaning rate occurs when the footprint of the sensor over the evader region is maximal. For a line shaped sensor of length $2r$ this happens when the entire length of the sensor is fully inside the evader region and it moves perpendicular to its orientation. The rate of sweeping when this happens has to be higher than the minimal expansion rate of the evader region (given its total area) otherwise no sweeping process can ensure detection of all evaders. We analyze the search process when the sweeper swarm is comprised of $n$ identical agents. The smallest searcher velocity satisfying this requirement is defined as the critical velocity and denoted by $V_{LB}$, we have:
\newtheorem{thm}{Theorem}
\begin{thm}
No sweeping process will be able to successfully complete the confinement task if its velocity, $V_s$, is less than,
\begin{equation}
{{\rm{V}}_{LB}} = \frac{{\pi {R_0}{V_T}}}{nr}
\label{e1}
\end{equation}
\end{thm}
\begin{proof}
Denote by $\Delta T$ the interval of search. The maximal area that can be scanned when the searcher moves with a velocity $V_s$ is given by,
\begin{equation}
{A_{Max Clean}} = 2r{V_s}\Delta T
\label{e2}
\end{equation}
Therefore if each agent in the swarm of $n$ cleans the maximal possible area, the maximal area that the swarm cleans is given by,
\begin{equation}
{A_{Max Clean}}(n) = 2nr{V_s}\Delta T
\label{e6}
\end{equation}
i.e., the best cleaning rate is $2rnV_s$. The least spread of the evader region that expands due to evaders' possible motion with velocity $V_T$ occurs when the region has the shape of a circle. This is due to the isoperimetric inequality: for a given area the minimal boundary length that encloses it happens when the shape of the region is circular. Therefore for an initial circular region with radius $R_0$ the evader region minimal expansion will be to a circle with a larger radius. For a spread of $\Delta T$ the radius of the evader region will grow to be $R_0 + \Delta T V_T$ and the area of the evader region will grow from $\pi {R_0}^2$ to $\pi {({R_0} + \Delta T{V_T})^2}$. Therefore the growth of the evader region area in time $\Delta T$ will be
${A_{Least Spread}} = \pi {({R_0} + \Delta T{V_T})^2} - \pi {R_0}^2 = 2\pi {R_0}\Delta T{V_T} + {\left( {\Delta T{V_T}} \right)^2}$. The spread rate will therefore be the division of the last expression by $\Delta T$. Letting $\Delta T \to 0$ the expression resolves to $2\pi {R_0}{V_T}$, the least possible spread rate. In order to guarantee the possibility of sweeping we must set the best cleaning rate to be larger than the worst spread of area that is $2rn{V_s} \geq 2\pi {R_0}{V_T}$. This yields the minimal velocity of a sweeper regardless of the search process it employs. Hence,
\begin{equation}
{V_s} \geq\frac{{\pi {R_0}{V_T}}}{nr}=
{{\rm{V}}_{LB}}
\label{e3}
\end{equation}
\end{proof}
Hopefully, after the first sweep the evader region is within a circle with a smaller radius than the initial evader region's radius. Since the sweepers travel along the perimeter of the evader region and this perimeter decreases after the first sweep, ensuring a sufficient sweeper velocity that guarantees that no evader escapes during the initial sweep guarantees also that the sweeper velocity is sufficient to prevent escape in subsequent sweeps as well.  The formulation of the problem in terms of the smallest possible searcher velocity that is needed in order to guarantee a no escape search is equivalent to asking what is the maximal boundable circular region that is possible to confine the evaders to given a searcher's velocity of $V_s$, sensor length of $2r$, $n$ sweepers and a maximal velocity of an evading agent that is equal to $V_T$. A plot of $V_{LB}$ is presented in Fig. $1$.
\begin{figure}[ht]
\noindent \centering{}\includegraphics[width=3.4in,height=2.8in]{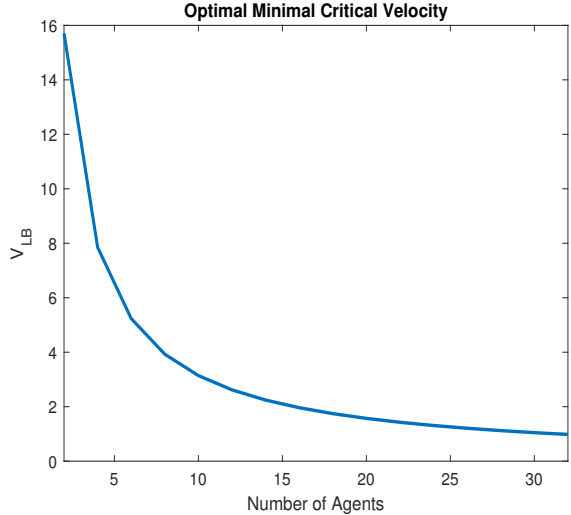} \caption{Minimal critical sweeper velocity. This velocity is a lower bound for any search process that the sweepers apply. $V_{LB}$ depends only on $R_0$, $r$, $V_T$ and $n$, the number of sweepers in the swarm. In this figure we plot the critical velocities for an even number of agents, ranging from $2$ to $32$ agents. The parameters values chosen for this plot are $r=10$, $V_T = 1$ and $R_0 = 100$.}
\label{Fig1Label}
\end{figure}

\section{Multiple Agents with Linear Sensors: the Circular Sweep Process}
A multi-agent swarm that is given the task of confining and detecting smart evaders that are initially located in a circular region with a given radius can be analyzed in a number of ways. One way is to examine the outcomes of adding multiple sweepers each equipped with a line sensor of length $2r$ and analyzing the times it takes to complete the detection of all evaders, which we will refer to as the cleaning of the evader region. In the single agent search problem described in \cite{francos2019search}, we observed that there can be escape from point $P=(0,R_0)$ when basing the searcher's velocity only on a single traversal around the evader region. Therefore we had to increase the agent's critical velocity to deal with this possible escape. If we were to distribute a multi-agent swarm  say, equally along the boundary of the initial evader region, we would have the same problem of possible escape from the points adjacent to the starting of the sweepers. We wish the sweepers to have the lowest possible critical velocity, hence we propose a different idea for the search process. The idea is to have pairs of sweeping agents move out in opposite directions along the boundary of the evader region and sweep in a pincer movement rather than having a convoy of sweepers moving in the same direction along the boundary. Our method is readily applicable for any even number of sweepers. Pairs of sweepers may start with half the length of their sensors inside the evader region while the other half is outside the region. The search process can be viewed as a $2$ dimensional search on which the actual agents travel on a plane or as a $3$ dimensional search were the sweepers are drone like agents which fly over the evader area. Evaders that intersect the sweepers sensors are detected. From any point in which they are located, evaders can move in any direction at a maximal velocity of $V_T$. The sweepers are positioned back to back with each other. One sweeper in the pair moves counter clockwise while the other sweeper in the pair moves clockwise. In case the search is planar, once the sweepers meet, i.e. their sensors are again back to back at another point, they switch directions and change the direction in which they move. That is, the agent that traveled counter clockwise will now travel clockwise and vice versa. For example, if the search if carried out by two sweepers, after the first sweep this switching point is located at $(0,-R_0)$. This changing of directions occurs every time a sweeper bumps into another. Each sweeper is responsible for an angular sector of the evader region that is proportional to the number of participating agents in the search.
In case the search is $3$ dimensional, where the sweepers fly at different heights above the evader region, every time a sweeper is directly above another, they exchange the angular section they are responsible to sweep between them and continue the search. The analysis of the two cases is exactly the same.
An illustration of the initial placement of agents for $2$ participating sweepers is presented in Fig. $2$.
\begin{figure}[ht]
\noindent \centering{}\includegraphics[width=3in,height =2.5in]{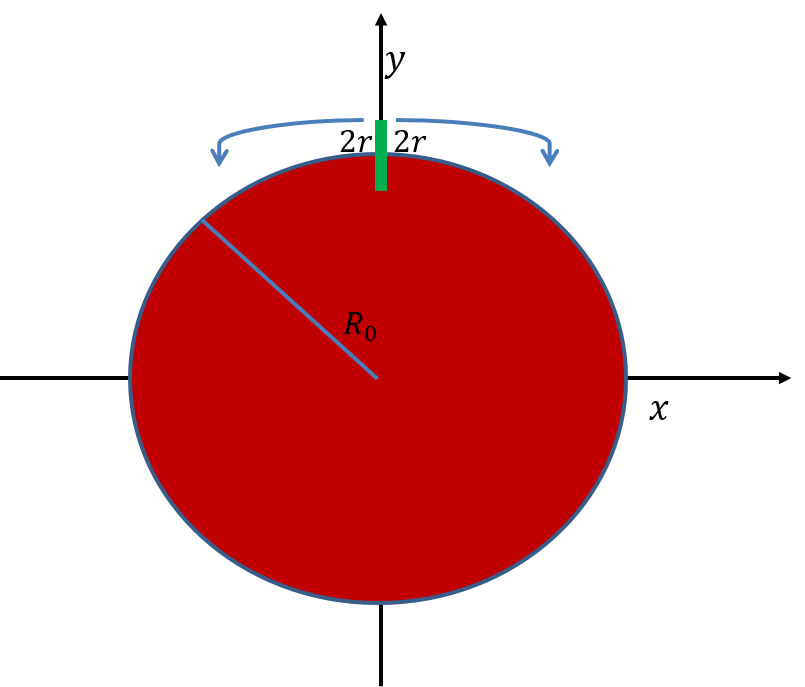} \caption{Initial placement of $2$ agents employing the circular sweep process.}
\label{Fig2Label}
\end{figure}
We analyze the case that the multi-agent swarm consists of $n$ agents, where $n$ is an even number, and each sweeper has a sensor length of $2r$. At the beginning of the search process the footprint of each sweeper's sensor that is over the evader region is equal to $r$. When employing this type of search pattern the symmetry between the two agent trajectories prevents the escape from point $P=(0,R_0)$ that is the most dangerous point an evader can escape from as proved in the single agent scenario described in \cite{francos2019search}. Therefore each sweeper's critical velocity can be based only upon the time it takes it to traverse the angular section it responsible for, namely $\frac{2\pi}{n}$. For example, in the scenario that the sweeper swarm consists of only two sweepers each sweeper is required to scan an angle of $\pi$. As is in the case of the single agent search, if the sweepers velocities are above the critical velocity of the scenario the agents can advance inwards towards the center of the evader region after completing a cycle. For the multi agent case the notion of a cycle or an iteration corresponds to an agent's traversal of the angular section it is required to scan. Once the agents finish scanning the angular section they are responsible for, and if their velocities are greater then the critical velocity that corresponds to the scenario they advance inwards together. For the planar search, only after the inward advancement do the sweepers change the direction of scanning. For the $3$ dimensional search the sweepers advance inwards together and after that exchange between them the angular section they are responsible to sweep and continue to scan a section with a smaller radius. Since each sweeper has a sensor length of $r$ outside the evader region, during an angular traversal of $\frac{2\pi}{n}$ around the evader region radius of $R_0$,  in order to guarantee that no point in the evader region escapes the sweepers we must demand that the spread of that point will be confined in a radius of no more than $r$ from the point it originated from at the beginning of the cycle. We therefore have that the following inequality must be satisfied,
\begin{equation}
\frac{{2\pi {R_0}}}{{n{V_s}}} \le \frac{r}{{{V_T}}}
\label{e22}
\end{equation}
Rearranging terms yields that the sweepers velocities must satisfy that,
\begin{equation}
{V_s} \ge \frac{{2\pi {R_0}{V_T}}}{{nr}}
\label{e23}
\end{equation}
The critical velocity for the circular sweep process is therefore given when we have equality in (\ref{e23}).
\begin{equation}
{V_c} = \frac{{2\pi {R_0}{V_T}}}{{nr}}
\label{e24}
\end{equation}
A comparison between the optimal minimal critical velocity and the circular critical velocity is presented in Fig. $3$. It can be noted from the terms in (\ref{e1}) and (\ref{e24}) that the circular critical velocity is exactly twice the optimal critical velocity.
\begin{figure}[ht]
\noindent \centering{}\includegraphics[width=3.4in,height=2.8in]{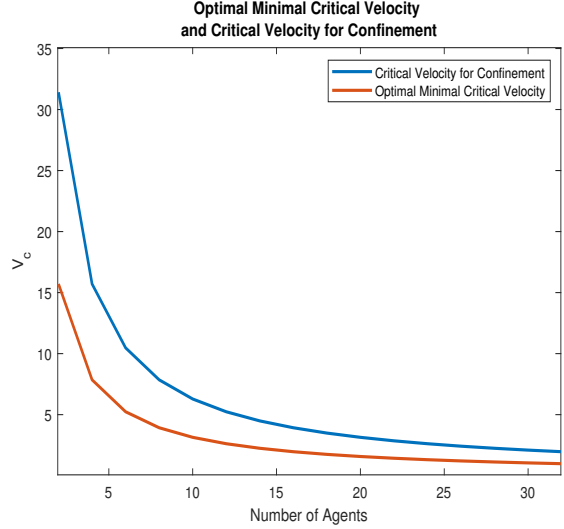} \caption{In this figure we plot the critical velocity as a function of the number of sweeper agents. The number of sweeper agents is even, and ranges from $2$ to $32$ agents, that employ the multi-agent circular sweep process. We also plot the optimal critical velocity for comparison.  The parameters values chosen for this plot are $r=10$, $V_T = 1$ and $R_0 = 100$.}
\label{Fig3Label}
\end{figure}

\begin{thm}
For an $n$ agent swarm for which $n$ is even that performs the circular sweep process where the sweeper distribution is as described, the number of iterations it will take the swarm to reduce the evader region to be bounded by a circle with a radius that is less than or equal to $r$ is given by,
\begin{equation}
{N_n} = \left\lceil {\frac{{\ln \left( {\frac{{2\pi r{V_T} - nr{V_s}}}{{2\pi {R_0}{V_T} - nr{V_s}}}} \right)}}{{\ln \left( {1 + \frac{{2\pi {V_T}}}{{n\left( {{V_s} + {V_T}} \right)}}} \right)}}} \right\rceil
\label{e98986}
\end{equation}
After ${N_n}$ sweeps the sweeper swarm performs an additional circular sweep and cleans the entire evader region.
\\
\\
We denote by ${T_{in}}$ the sum of all inward advancement times and by ${T_{circular}}$ the sum of all the circular traversal times. Therefore the time it takes the swarm to clean the entire evader region is given by,
\begin{equation}
T(n) = {T_{in}}(n) + {T_{circular}}(n)
\label{e98988}
\end{equation}
Where ${T_{in}}(n)$ is given by,
\begin{equation}
{T_{in}}(n) = \frac{{{R_0}}}{{{V_s}}} + \left( {\frac{{2\pi {R_0}{V_T} - nr{V_s}}}{{n{V_s}\left( {{V_s} + {V_T}} \right)}}} \right){\left( {1 + \frac{{2\pi {V_T}}}{{n\left( {{V_s} + {V_T}} \right)}}} \right)^{{N_n} - 1}}
\label{e98989}
\end{equation}
And ${T_{circular}}(n)$ is given by,
\begin{equation}
\begin{array}{l}
{T_{circular}}(n) =  - \frac{{{R_0}\left( {{V_s} + {V_T}} \right)}}{{{V_T}{V_s}}} + \frac{{nr\left( {{V_s} + {V_T}} \right) + 2\pi r{V_T}}}{{2\pi {V_T}^2}}
 + {\left( {1 + \frac{{2\pi {V_T}}}{{n\left( {{V_s} + {V_T}} \right)}}} \right)^{{N_n}}}\left( {\frac{{\left( {{V_s} + {V_T}} \right)\left( {2\pi {R_0}{V_T} - rn{V_s}} \right)}}{{2\pi {V_s}{V_T}^2}}} \right)
 + \frac{{r\left( {{N_n} - 1} \right)}}{{{V_T}}} + \frac{{2\pi r}}{{n{V_s}}}
\end{array}
\label{e98990}
\end{equation}
\end{thm}
\begin{proof}
Let us denote by $\Delta V >0$ the addition to the sweeper's velocity above the critical velocity. The sweeper's velocity is therefore given by, $V_s = V_c + \Delta V$. The time it takes each sweeper to circularly sweep the region it is responsible to sweep is given by,
\begin{equation}
{T_{circular}}_i = \frac{{2\pi {R_i}}}{{n({V_c} + \Delta V)}}
\label{e25}
\end{equation}
Since $V_s = V_c + \Delta V$, ${T_{circular}}_i$ can also be expressed as,
\begin{equation}
{T_{circular}}_i = \frac{{2\pi {R_i}}}{{n{V_s}}}
\label{e20}
\end{equation}
Depending on the number of sweepers and the iteration number we have that the distance a sweeper can advance inwards after completing an iteration is given by,
\begin{equation}
{\delta _i}(\Delta V) = r - {V_T}{T_{circular}}_i
\label{e1090}
\end{equation}
Where in the term ${\delta _i}(\Delta V)$, $\Delta V$ denotes the increase in the agent velocity relative to the critical velocity, and $i$ denotes the number of sweep iterations the sweeper performed around the evader region, where $i$ starts from sweep number $0$. Since a sweeper cannot advance after each iteration by a distance that is larger than its sensor length and still prevent the escape of an evader with an arbitrary trajectory, ${\delta _i}(\Delta V)$ is bounded between,
\begin{equation}
0 \le {\delta _i}(\Delta V) \le r
\label{e26}
\end{equation}
The time it takes the sweepers to move inwards until half of their sensors are over the evader region depends on the relative velocity between the agents inwards entry and the evader region outwards expansion and is given by (\ref{e1068}). Therefore the distance an agent can advance inwards after completing an iteration is given by,
\begin{equation}
{\delta _{{i_{eff}}}}(\Delta V) = {\delta _i}(\Delta V)\left( {\frac{{{V_s}}}{{{V_s} + {V_T}}}} \right)
\label{e500}
\end{equation}
The inward advancement time depends on the iteration number. It is denoted by $T_{i{n_i}}$ and is given by,
\begin{equation}
{T_{i{n_i}}} = \frac{{{\delta _{{i_{eff}}}}(\Delta V)}}{{{V_s}}} = \frac{{rn{V_s} - 2\pi {R_i}{V_T}}}{{n{V_s}\left( {{V_s} + {V_T}} \right)}}
\label{e1068}
\end{equation}
Where the index $i$ in ${T_{i{n_i}}}$ denotes the iteration number in which the advancement is done. After all agents complete their sweep the evader region is bounded by a circle with a smaller radius compared to the previous sweep. Thus the new radius of the circle that will bound the evader region is given by,
\begin{equation}
{R_{i + 1}} = {R_i} - {\delta _{{i_{eff}}}}(\Delta V) = {R_i} - {\delta _i}(\Delta V)\left( {\frac{{{V_s}}}{{{V_s} + {V_T}}}} \right)
\label{e1069}
\end{equation}
Plugging the value of ${\delta _i}(\Delta V)$ from (\ref{e1090}) into (\ref{e1069}) yields,
\begin{equation}
{R_{i + 1}} = {R_i} - {\delta _{{i_{eff}}}}(\Delta V) = {R_i} - \frac{{r{V_s}}}{{{V_s} + {V_T}}} + \frac{{2\pi {R_i}{V_T}}}{{n\left( {{V_s} + {V_T}} \right)}}
\label{e1093}
\end{equation}
Rearranging terms yields,
\begin{equation}
{R_{i + 1}} = {R_i}\left( {1 + \frac{{2\pi {V_T}}}{{n\left( {{V_s} + {V_T}} \right)}}} \right) - \frac{{r{V_s}}}{{{V_s} + {V_T}}}
\label{e1094}
\end{equation}
For any number of even sweepers, $n$, the search continues in this way until the evader region is confined to a radius of $\widehat{{R_N}} = r$.
Denoting the coefficients $c_1$ and $c_3$ by,
\begin{equation}
{c_3} = 1 + \frac{{2\pi {V_T}}}{{n\left( {{V_s} + {V_T}} \right)}},{c_1} =  - \frac{{r{V_s}}}{{{V_s} + {V_T}}}
\label{e1095}
\end{equation}
Thus (\ref{e1094}) takes the form of,
\begin{equation}
{R_{i + 1}} = {c_3}{R_i} + {c_1}
\label{e1300}
\end{equation}
The number of iterations it takes the sweeper swarm to reduce the evader region to be bounded by a circle with a radius of $\widehat{{R_N}}=r$, that corresponds to the last sweep before completely cleaning the evader region is calculated in Appendix $A$. It is given by,
\begin{equation}
N_n = \left\lceil {\frac{{\ln \left( {\frac{{\widehat{{R_N}} - \frac{{{c_1}}}{{1 - {c_3}}}}}{{{R_0} - \frac{{{c_1}}}{{1 - {c_3}}}}}} \right)}}{{\ln {c_3}}}} \right\rceil
\label{e11}
\end{equation}
Substitution of coefficients in (\ref{e11}) yields that the number of iterations it takes the sweepers to reduce the evader region to be contained in a circle with the radius of the last scan, $\widehat{R_N} = r$, is given by,
\begin{equation}
{N_n} = \left\lceil {\frac{{\ln \left( {\frac{{2\pi r{V_T} - nr{V_s}}}{{2\pi {R_0}{V_T} - nr{V_s}}}} \right)}}{{\ln \left( {1 + \frac{{2\pi {V_T}}}{{n\left( {{V_s} + {V_T}} \right)}}} \right)}}} \right\rceil
\label{e1019}
\end{equation}
The total time it takes the  multi-agent swarm of $n$ sweepers to clean the evader region is given by total time of inward advancements combined with the times it takes the sweepers to complete the circular traversal of the evader region in all cycles. We denote by ${T_{in}}(n)$ the sum of all the inward advancement times and by ${T_{circular}}(n)$ the sum of all the circular traversal times. Namely we have that,
\begin{equation}
T(n) = {T_{in}}(n) + {T_{circular}}(n)
\label{e1101}
\end{equation}
We denote the total advancement time until the evader region is bounded by a circle with a radius that is less than or equal to $r$ as $\widetilde{{T_{in}}}(n)$. It is given by,
\begin{equation}
\widetilde{{T_{in}}}(n) = \sum\limits_{i = 0}^{{N_n}-2} {{T_{i{n_i}}}}
\label{e1099}
\end{equation}
During the inward advancements only the tip of the sensor, that has zero width, is inserted into the evader region. Therefore it does not detect any evaders until it completes its inward advance and starts sweeping again. After the sweeper completed its advance into the evader region its sensor footprint over the evader region is equal to $r$. The total search time until the evader region is bounded by a circle with a radius that is less than or equal to $r$ is given by the sum of the total circular sweep times and the times of the inward advances. Namely,
\begin{equation}
\widetilde{T}(n) = \widetilde{{T_{in}}}(n) + \widetilde{{T_{circular}}}(n)
\label{e1081}
\end{equation}
Using the developed term for $T_{i{n_i}}$ the total inward advancement times until the evader region is bounded by a circle with a radius that is less than or equal to $r$ are computed by,
\begin{equation}
\widetilde{{T_{in}}}(n) = \sum\limits_{i = 0}^{{N_n} - 2} {{T_{i{n_i}}} = } \frac{{\left( {{N_n} - 1} \right)r}}{{{V_s} + {V_T}}} - \frac{{2\pi {V_T}\sum\limits_{i = 0}^{{N_n} - 2} {{R_i}} }}{{n{V_s}\left( {{V_s} + {V_T}} \right)}}
\label{e1020}
\end{equation}
We note that the first inward advancement occurs when the evader region is bounded by a circle of radius $R_0$ and the last inward advancement occurs at iteration number ${N_n}-2$, which describes the inward advancement in which the evader region transitions from being bounded by a circle of radius ${R_{{N_n} - 2}}$ to being bounded by a circle of radius ${R_{{N_n} - 1}}$. Afterwards the sweeper swarm completes another circular sweep where after its completion the evader region is bounded by a circle of radius $R_N$. The calculation is done in this way since at the last sweep the sweeping agents advance a distance that is equal to or smaller than the allowable distance they can advance towards the center of the evader region. This occurs since we don't want the sweepers paths to cross each other. We desire that the lower tips of the sweepers' sensors will not cross the center of the evader region in order to prevent collisions between the sweeping agents at the last iteration before they completely clean the evader region. The full derivation of $\widetilde{{T_{in}}}(n)$ can be found in Appendix $F$. This derivation yields that,
\begin{equation}
\widetilde{{T_{in}}}(n) = \sum\limits_{i = 0}^{{N_n} - 2} {{T_{i{n_i}}}}  = \frac{{{R_0}}}{{{V_s}}} - \frac{{nr}}{{2\pi {V_T}}}
 - {\left( {1 + \frac{{2\pi {V_T}}}{{n\left( {{V_s} + {V_T}} \right)}}} \right)^{{N_n} - 1}}\left( {\frac{{2\pi {R_0}{V_T} - nr{V_s}}}{{2\pi {V_T}{V_s}}}} \right)
\label{e1100}
\end{equation}
In order to calculate ${T_{in}}(n)$ we must add the last inward advancement. This time is given by
\begin{equation}
{T_{_{in}last}}(n)=\frac{{{R_{{N_n}}}}}{{{V_s}}}
\label{e505}
\end{equation}
Therefore,
\begin{equation}
{T_{_{in}last}}(n)= \frac{{nr}}{{2\pi {V_T}}} + {\left( {1 + \frac{{2\pi {V_T}}}{{n\left( {{V_s} + {V_T}} \right)}}} \right)^{{N_n}}}\left( {\frac{{2\pi {R_0}{V_T} - nr{V_s}}}{{2\pi {V_T}{V_s}}}} \right)
\label{e506}
\end{equation}
${T_{in}}(n)$ is given as ${T_{in}}(n) = \widetilde{{T_{in}}}(n) + {T_{_{in}last}}(n)$ and therefore yields,
\begin{equation}
{T_{in}}(n) = \frac{{{R_0}}}{{{V_s}}} + \left( {\frac{{2\pi {R_0}{V_T} - nr{V_s}}}{{n{V_s}\left( {{V_s} + {V_T}} \right)}}} \right){\left( {1 + \frac{{2\pi {V_T}}}{{n\left( {{V_s} + {V_T}} \right)}}} \right)^{{N_n} - 1}}
\label{e507}
\end{equation}
We now proceed to the calculation of the circular sweep times. The initial circular sweep time is given by,
\begin{equation}
{T_0} = \frac{{2\pi {R_0}}}{{n{V_s}}}
\label{e31}
\end{equation}
The relation between the time to circularly sweep a circle of radius $R_i$ by an angle of $\frac{2\pi}{n}$ at a velocity of $V_s$ is given by,
\begin{equation}
{T_i} = \frac{{2\pi {R_i}}}{{n{V_s}}}
\label{e32}
\end{equation}
We denote the coefficient $c_4$ by,
\begin{equation}
{c_4} =  - \frac{{2\pi r}}{{n\left( {{V_s} + {V_T}} \right)}}
\label{e30}
\end{equation}
It can be noted that by multiplying (\ref{e1300}) by $\frac{2\pi}{n{V_s}}$ we obtain a recursive difference equation for the sweep times. Therefore the sweep times can be written as,
\begin{equation}
{T_{i + 1}} = {c_3}{T_i} + {c_4}
\label{e12}
\end{equation}
Where in this context each sweep iteration is defined as a traversal of an angle of $\frac{2\pi}{n}$ by the sweeper.
We denote the sum of circular sweep times until the evader region is bounded by a circle that is less than or equal to $r$ by $\widetilde{{T_{circular}}}$.  $\widetilde{{T_{circular}}}$ is developed in Appendix $C$ and is given by,
\begin{equation}
\widetilde{T_{circular}}(n) = \frac{{{T_0} - {c_3}{T_{N_n - 1}} + \left( {{N_n} - 1} \right){c_4}}}{{1 - {c_3}}}
\label{e8}
\end{equation}
The last circular sweep time before the evader region is bounded by a circle with a radius that is smaller or equal to $r$ is computed in Appendix $D$ and is given by,
\begin{equation}
{T_{{N_n} - 1}} = \frac{{{c_4}}}{{1 - {c_3}}} +  {{c_3}^{{N_n} - 1}}\left( {{T_0} - \frac{{{c_4}}}{{1 - {c_3}}}} \right)
\label{e10}
\end{equation}
Plugging the respective coefficients into (\ref{e10}) yields,
\begin{equation}
{T_{{N_n} - 1}} = \frac{r}{{{V_T}}} + {\left( {1 + \frac{{2\pi {V_T}}}{{n\left( {{V_s} + {V_T}} \right)}}} \right)^{{N_n} - 1}}\left( {\frac{{2\pi {R_0}{V_T} - rn{V_s}}}{{n{V_s}{V_T}}}} \right)
\label{e35}
\end{equation}
Substituting the coefficients in (\ref{e8}) with the respective developed terms yields,
\begin{equation}
\begin{array}{l}
\widetilde{{T_{circular}}} =  - \frac{{{R_0}\left( {{V_s} + {V_T}} \right)}}{{{V_T}{V_s}}} + \frac{{nr\left( {{V_s} + {V_T}} \right) + 2\pi r{V_T}}}{{2\pi {V_T}^2}}
 + {\left( {1 + \frac{{2\pi {V_T}}}{{n\left( {{V_s} + {V_T}} \right)}}} \right)^{{N_n}}}\left( {\frac{{\left( {{V_s} + {V_T}} \right)\left( {2\pi {R_0}{V_T} - rn{V_s}} \right)}}{{2\pi {V_s}{V_T}^2}}} \right)
 + \frac{{r\left( {{N_n} - 1} \right)}}{{{V_T}}}
\end{array}
\label{e1035}
\end{equation}
After the completion of sweep $N_n$ the evader region is bounded by a circle with a radius that is less than or equal to $r$. In order to prevent the prevent the paths of the sweepers from coinciding at the last sweep, the sweepers advance towards the center of the evader region until the lower tips of their sensors are at the center of the evader region. Following this advancement they perform the last circular sweep. The time to perform this sweep is denoted by ${T_{last}}(n)$. ${T_{last}}(n)$ is the time it takes the sweepers to complete the last circular sweep of radius $r$ while traversing an angle of $\frac{2\pi}{n}$ around the center of the evader region. ${T_{last}}(n)$ is given by,
\begin{equation}
{T_{last}}(n) = \frac{{2\pi r}}{{n{V_s}}}
\label{e16}
\end{equation}
Therefore the total time of circular sweeps until complete cleaning of the evader region is given by,
\begin{equation}
{T_{circular}}(n) = \widetilde{{T_{circular}}}(n) + {T_{last}}(n)
\label{e723}
\end{equation}
\begin{equation}
\begin{array}{l}
{T_{circular}}(n) =  - \frac{{{R_0}\left( {{V_s} + {V_T}} \right)}}{{{V_T}{V_s}}} + \frac{{nr\left( {{V_s} + {V_T}} \right) + 2\pi r{V_T}}}{{2\pi {V_T}^2}}
 + {\left( {1 + \frac{{2\pi {V_T}}}{{n\left( {{V_s} + {V_T}} \right)}}} \right)^{{N_n}}}\left( {\frac{{\left( {{V_s} + {V_T}} \right)\left( {2\pi {R_0}{V_T} - rn{V_s}} \right)}}{{2\pi {V_s}{V_T}^2}}} \right)
 + \frac{{r\left( {{N_n} - 1} \right)}}{{{V_T}}} + \frac{{2\pi r}}{{n{V_s}}}
\end{array}
\label{e724}
\end{equation}
\end{proof}
\begin{figure}[ht]
\noindent \centering{}\includegraphics[width=3.4in,height=2.8in]{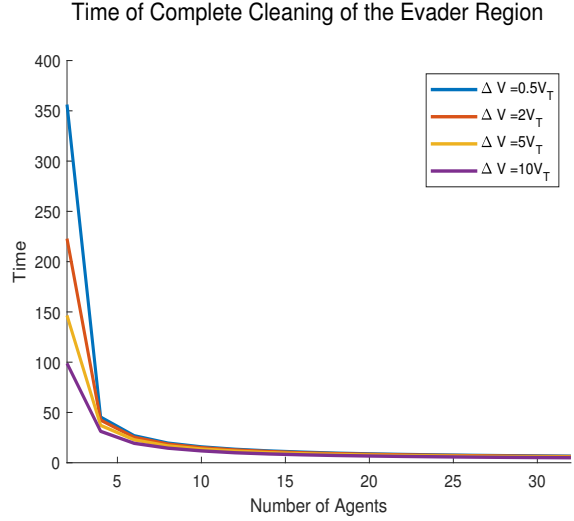} \caption{Time of complete cleaning of the evader region. In this figure we simulated the circular sweep processes for an even number of agents, ranging from $2$ to $32$ agents, that employ the multi-agent circular sweep process. The parameters values chosen for this plot are $r=10$, $V_T = 1$ and $R_0 = 100$.}
\label{Fig4Label}
\end{figure}

\begin{figure}[ht]
\noindent \centering{}\includegraphics[width=3.4in,height=2.8in]{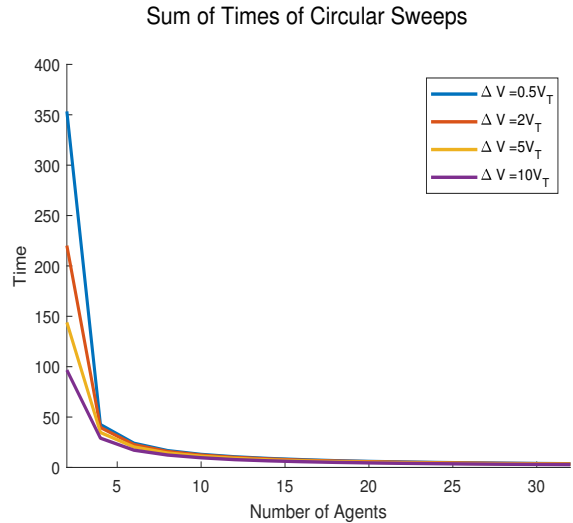} \caption{Sum of the circular sweep times of the search until complete cleaning of the evader region. In this figure we simulated the circular sweep processes for an even number of agents, ranging from $2$ to $32$ agents, that employ the multi-agent circular sweep process. The parameters values chosen for this plot are $r=10$, $V_T = 1$ and $R_0 = 100$.}
\label{Fig5Label}
\end{figure}

\begin{figure}[ht]
\noindent \centering{}\includegraphics[width=3.4in,height=2.8in]{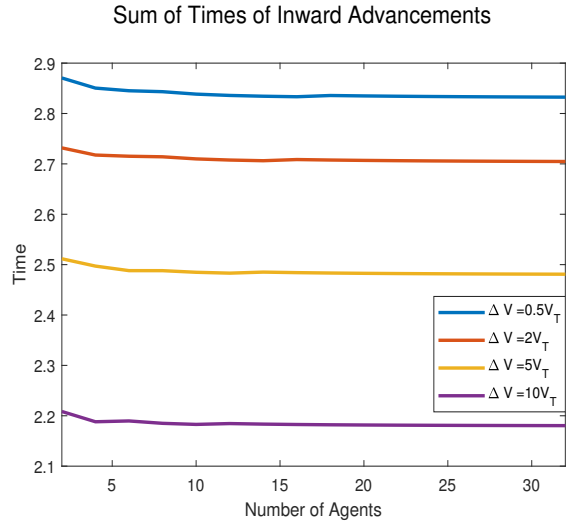} \caption{Sum of the inward advancement times until complete cleaning of the evader region. In this figure we simulated the circular sweep processes for an even number of agents, ranging from $2$ to $32$ agents, that employ the multi-agent circular sweep process. The parameters values chosen for this plot are $r=10$, $V_T = 1$ and $R_0 = 100$.}
\label{Fig6Label}
\end{figure}

\begin{figure}[ht]
\noindent \centering{}\includegraphics[width=3.4in,height=2.8in]{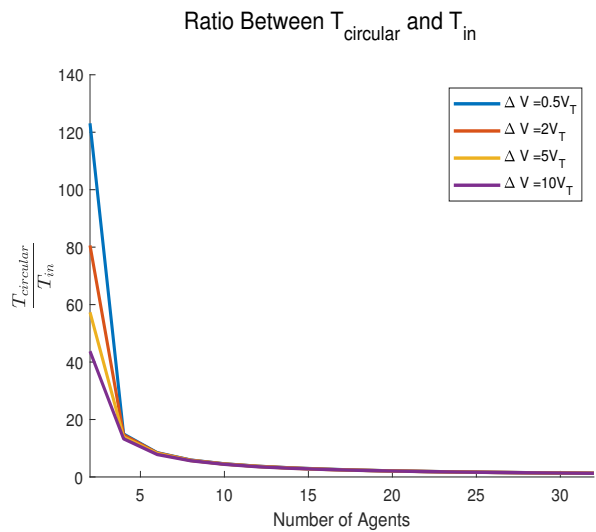} \caption{Ratio between the circular sweep times of the search and the inward advancement times until complete cleaning of the evader region. In this figure we simulated the circular sweep processes for an even number of agents, ranging from $2$ to $32$ agents, that employ the multi-agent circular sweep process. The parameters values chosen for this plot are $r=10$, $V_T = 1$ and $R_0 = 100$.}
\label{Fig7Label}
\end{figure}

\begin{figure}[ht]
\noindent \centering{}\includegraphics[width=3.4in,height=2.8in]{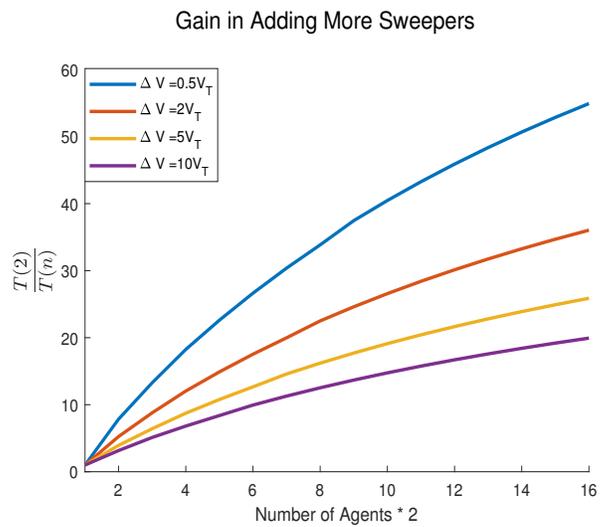} \caption{Gain in cleaning time obtained by adding more sweepers. In this figure we simulated the circular sweep processes for an even number of agents, denoted by $n$, ranging from $2$ to $32$ agents that employ the multi-agent circular sweep process. In each of the curves, every point is obtained by the ratio between the sweep times of a $2$ agent swarm and an $n$ agent sweeper swarm. We show the results obtained for different values of velocities above the circular critical velocity, i.e. different choices for $\Delta V$. The parameters values chosen for this plot are $r=10$, $V_T = 1$ and $R_0 = 100$.}
\label{Fig8Label}
\end{figure}
\section{Multiple Agents with Linear Sensors: the Spiral Sweep Process}
Since at the start of every circular sweep process half of the sweepers' sensors are outside of the evader region we would like the sweepers to employ a more efficient motion throughout the cleaning process. This means that throughout the motion of the searcher the footprint of its sensor that is above the evader region will be maximal. This can be achieved with a spiral scan, where the agent sensor tracks the expanding evader region wavefront, while preserving its shape to be as close as possible to a circle. An illustration of the initial placement of $2$ agents that employ the spiral sweep process is presented in Fig. $9$.
\begin{figure}[ht]
\noindent \centering{}\includegraphics[width=3in,height =2.5in]{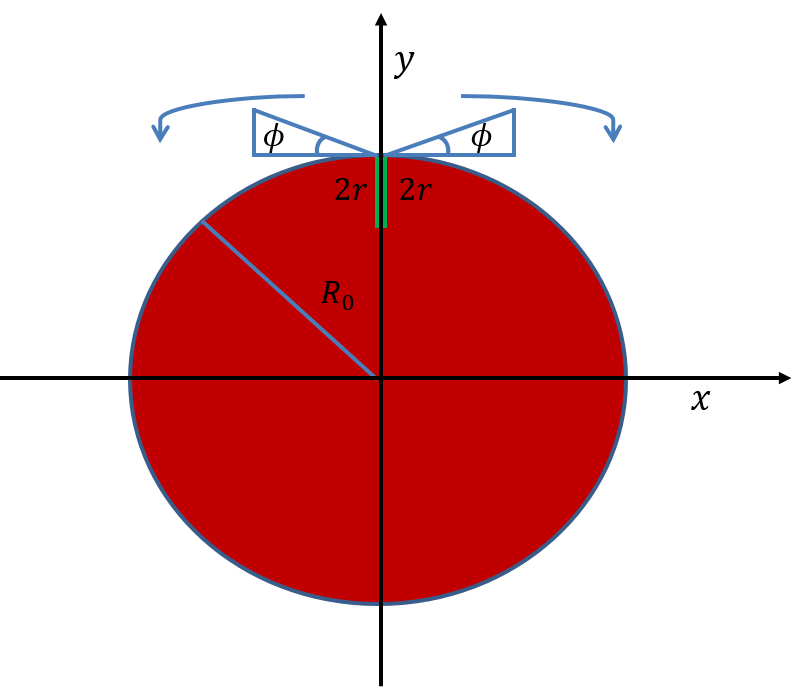} \caption{Initial placement of $2$ agents employing the spiral sweep process.}
\label{Fig9Label}
\end{figure}
Each sweeper in the swarm has a line sensor of length $2r$. We choose that all sweepers will have sensors of equal lengths. This choice implies that when the sweepers sensors reach the same point and are tangent to each other, escape will not be possible from the gap between the sensors. Such a gap is prevented when using sensors of equal length between all agents. This is a necessary requirement since the process that is described below relies on the symmetry between the agents sensors that are over the evader region. In the single agent search problem described in \cite{francos2019search}, we observed that there can be escape from point $P=(0,R_0)$ when basing the searcher's velocity only on a single traversal around the evader region. Therefore we had to increase the agent's critical velocity to deal with this possible escape. If we were to distribute a multi-agent swarm say, equally along the boundary of the initial evader region, we would have the same problem of possible escape from the points adjacent to the starting of the sweepers. We wish the sweepers to have the lowest possible critical velocity, hence we propose a different idea for the search process. The idea is to have pairs of sweeping agents move out in opposite directions along the boundary of the evader region and sweep in a pincer movement rather than having a convoy of sweepers moving in the same direction along the boundary. Our method is readily applicable for any even number of sweepers. Pairs of sweepers start with the entire length of their sensors inside the evader region. The search process can be viewed as a $2$ dimensional search on which the actual agents travel on a plane or as a $3$ dimensional search were the sweepers are drone like agents which fly over the evader area. Evaders that intersect the sweepers sensors are detected. From any point in which they are located, evaders can move in any direction at a maximal velocity of $V_T$. The sweepers are positioned back to back with each other. One sweeper in the pair moves counter clockwise while the other sweeper in the pair moves clockwise. In case the search is planar, once the sweepers meet, i.e. their sensors are again back to back at another point, they switch directions and change the direction in which they move. That is, the agent that traveled counter clockwise will now travel clockwise and vice versa. For example, if the search if carried out by two sweepers, after the first sweep this switching point is located at $(0,-R_0)$. This changing of directions occurs every time a sweeper bumps into another. Each sweeper is responsible for an angular sector of the evader region that is proportional to the number of participating agents in the search. In case the search is $3$ dimensional, where the sweepers fly at different heights above the evader region, every time a sweeper is directly above another, they exchange the angular section they are responsible to sweep between them and continue the search. The analysis of the two cases is exactly the same. We choose the sweepers sensors to have equal length in order to benefit from the symmetry implied by the trajectory of the pair. The combination of the trajectory and equal sensors lengths ensures that when the agents sensors reach the same point there will not be escape from the gap between the sensors. We analyze the case that the multi-agent swarm consists of $n$ agents, where $n$ is an even number, and each sweeper has a sensor length of $2r$. At the beginning of the search process the footprint of each sweeper's sensor that is over the evader region is equal to $2r$. When employing this type of search pattern the symmetry between the two agent trajectories prevents the escape from point $P=(0,R_0)$ that is the most dangerous point an evader can escape from as proved in the single agent scenario described in \cite{francos2019search}. Therefore each sweeper's critical velocity can be based only upon the time it takes it to traverse the angular section it responsible for, namely $\frac{2\pi}{n}$. For example, in the scenario that the sweeper swarm consists of only two sweepers each sweeper is required to scan an angle of $\pi$. As is in the case of the single agent search, if the sweepers velocities are above the critical velocity of the scenario the agents can advance inwards towards the center of the evader region after completing a cycle. For the multi agent case the notion of a cycle or an iteration corresponds to an agent's traversal of the angular section it is required to scan. Once the agents finish scanning the angular section they are responsible for, and if their velocities are greater then the critical velocity that corresponds to the scenario they advance inwards together. For the planar search, only after the inward advancement do the sweepers change the direction of scanning. For the $3$ dimensional search the sweepers advance inwards together and after that exchange between them the angular section they are responsible to sweep and continue to scan a section with a smaller radius. Each searcher begins its spiral traversal with the tip of its sensor tangent to the edge of the evader region. In order to keep its sensor tangent to the evader region, the searcher must travel at angle $\phi$ to the normal of the evader region. $\phi$ is calculated by,
\begin{equation}
\sin \phi  = \frac{{{V_T}}}{{{V_s}}}
\label{e1000}
\end{equation}
Thus we have,
\begin{equation}
\phi  = \arcsin \left( {\frac{{{V_T}}}{{{V_s}}}} \right)
\label{e1001}
\end{equation}
This method of traveling at angle $\phi$ preserves the evader region circular shape. Since the agent travels along the perimeter of the evader region and due to isoperimeteric inequality that states that for a given area the shape of the curve that bounds this area which will have the smallest perimeter is the circle this method will ensure that the time it takes to complete a sweep around the evader region will be minimal. The agent's angular velocity or rate of change of its angle with respect to the center of the evader region,$\theta _s$, can be described as a function of $\phi$ as,
\begin{equation}
\frac{{d{\theta _s}}}{{dt}} = \frac{{{V_s}\cos \phi }}{{{R_s}(t)}} = \frac{{\sqrt {{V_s}^2 - {V_T}^2} }}{{{R_s}(t)}}
\label{e1002}
\end{equation}
And the instantaneous growth rate of the searcher radius will be given by,
\begin{equation}
\frac{{d{R_s}(t)}}{{dt}} = {V_S}\sin \phi  = {V_T}
\label{e1003}
\end{equation}
Integrating equation (\ref{e1002}) between the initial and final sweep times of the angular section yields,
\begin{equation}
\int_0^{{t_\theta }} {\dot \theta \left( \zeta  \right)} d\zeta  = \int_0^{{t_\theta }} {\frac{{\sqrt {{V_s}^2 - {V_T}^2} }}{{{V_T}\zeta  + {R_0} - r}}d} \zeta
\label{e36}
\end{equation}
The result of the integral in (\ref{e36}) yields,
\begin{equation}
\theta \left( {{t_\theta }} \right) = \frac{{\sqrt {{V_s}^2 - {V_T}^2} }}{{{V_T}}}\ln \left( {\frac{{{V_T}{t_\theta } + {R_0} - r}}{{{R_0} - r}}} \right)
\label{e15}
\end{equation}
Applying the exponent function to both sides of the equation results in,
\begin{equation}
\left( {{R_0} - r} \right){e^{\frac{{{V_T}}}{{\sqrt {{V_s}^2 - {V_T}^2} }}}} = {V_T}t + {R_0} - r = {R_s}(t)
\label{e38}
\end{equation}
Each sweeper begins its spiral traversal with the tip of its sensor tangent to the edge of the evader region at point $P =(0,R_0)$. The time it takes the sweeper to complete a spiral traversal around the angular region of the evader region it is responsible to scan corresponds to changing its angle $\theta$ by $\frac{2\pi}{n}$. During this time the expansion of the evader region has to be by no more than $2r$ from its initial radius, in order for the sweeper to prevent the escape of potential evaders. This assertion holds under the assumption that after each cycle when the sweeper advances inwards toward the center of the evader region it completes this motion in zero time. Otherwise the spread of evaders has be less than $2r$ and considerations such as the spread of evaders during the inwards motion needs to be taken into account. This case will be addressed after the analysis of the simplified case that is described here. In order for no evader to escape the sweepers after a traversal of $\frac{2\pi}{n}$ the following inequality must hold,
\begin{equation}
{R_0} + r \ge {R_s}(t)
\label{e39}
\end{equation}
Substituting ${R_s}(t)$ with the expression of the trajectory of the center of the sweeper yields,
\begin{equation}
{R_0} + r \ge \left( {{R_0} - r} \right){e^{\frac{{2\pi {V_T}}}{{n\sqrt {{V_s}^2 - {V_T}^2} }}}}
\label{e42}
\end{equation}
Therefore in order to have a no escape sweep process the sweepers' velocities have to satisfy,
\begin{equation}
{V_S} \ge {V_T}\sqrt {\frac{{{{\left( {\frac{{2\pi }}{n}} \right)}^2}}}{{{{\left( {\ln \left( {\frac{{{R_0} + r}}{{{R_0} - r}}} \right)} \right)}^2}}} + 1}
\label{e43}
\end{equation}
We now propose a modification to the construction of the critical velocity given in (\ref{e43}). This modification takes into account the consideration that when the sweepers travel towards the center of the evader region after completing the spiral sweep they have to meet the evader wavefront travelling outwards the region with a speed of $V_T$ at the previous radius $R_0$. This more realistic update of the search process makes the spiral sweep process critical velocity agree with the optimal lower bound on the sweeper velocity that is independent of the sweep process and is slightly above it. We have that the expansion of the evader region during the first sweep denoted by ${T_c}$, has to satisfy that,
\begin{equation}
{V_T}{T_c} \leq \frac{{2r{V_s}}}{{{V_s} + {V_T}}}
\label{e48}
\end{equation}
Substituting the expression for ${T_c}$, yields
\begin{equation}
\left( {{R_0} - r} \right)\left( {{e^{\frac{{2\pi {V_T}}}{{n\sqrt {{V_s}^2 - {V_T}^2} }}}} - 1} \right) = \frac{{2r{V_s}}}{{{V_s} + {V_T}}}
\label{e49}
\end{equation}
These considerations are formulated in Corollary $1$.
\newtheorem{corollary}{Corollary}
\begin{corollary}
In the spiral sweep process for the $n$ agent case, where $n$ is even, the critical velocity ,${V_s}$, that allows the satisfaction of the confinement task, is obtained as the solution of,
\begin{equation}
{V_T}{T_c} = \frac{{2r{V_s}}}{{V_s} + {V_T}}
\label{e98996}
\end{equation}
Where ${T_c}$ is given by
\begin{equation}
{T_c} = \frac{{\left( {{R_0} - r} \right)\left( {{e^{\frac{{2\pi {V_T}}}{{n\sqrt {{V_s}^2 - {V_T}^2} }}}} - 1} \right)}}{{{V_T}}}
\label{e98997}
\end{equation}
\end{corollary}

\begin{figure}[ht]
\noindent \centering{}\includegraphics[width=3.4in,height=2.8in]{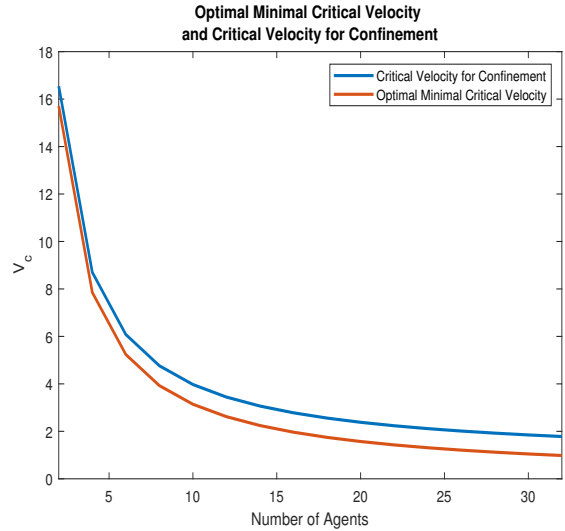} \caption{In this figure we plot the critical velocity as a function of the number of sweepers. The number of sweeper agents is even, and ranges from $2$ to $32$ agents, that employ the multi-agent spiral sweep process where the inward advancements towards the center of the evader are taken into account. We also plot the optimal critical velocity for comparison.  The parameters values chosen for this plot are $r=10$, $V_T = 1$ and $R_0 = 100$.}
\label{Fig10Label}
\end{figure}

\begin{figure}[ht]
\noindent \centering{}\includegraphics[width=3.4in,height=2.8in]{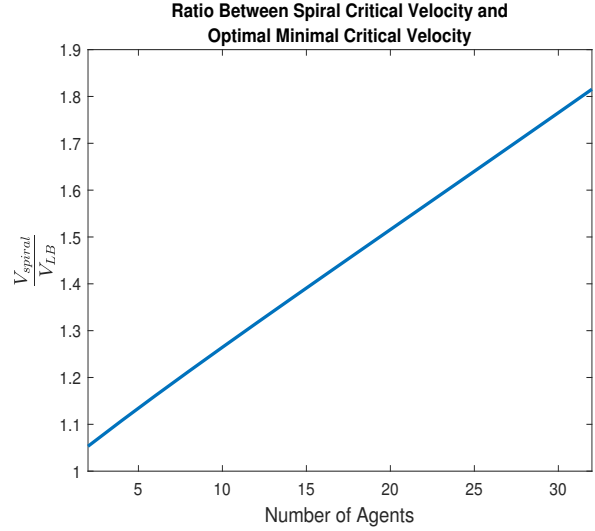} \caption{In this figure we plot the ratio between the spiral critical velocity and the optimal critical velocity as a function of the number of sweepers. The number of sweeper agents is even, and ranges from $2$ to $32$ agents, that employ the multi-agent spiral sweep process where the inward advancements towards the center of the evader are taken into account.  The parameters values chosen for this plot are $r=10$, $V_T = 1$ and $R_0 = 100$.}
\label{Fig11Label}
\end{figure}

\begin{thm}
For an $n$ agent swarm for which $n$ is even, that performs the spiral sweep process, where the sweepers distribution is as described, the number of iterations it will take the swarm to clean the entire evader region is given by,
\begin{equation}
\begin{array}{l}
\widetilde{{N_n}} = {N_n} + \eta + 1   = \\
\left\lceil {\frac{{\ln \left( {\frac{{r\left( {3 - {e^{\frac{{2\pi {V_T}}}{{n\sqrt {{V_s}^2 - {V_T}^2} }}}}} \right)}}{{{R_0}\left( {1 - {e^{\frac{{2\pi {V_T}}}{{n\sqrt {{V_s}^2 - {V_T}^2} }}}}} \right) + r\left( {1 + {e^{\frac{{2\pi {V_T}}}{{n\sqrt {{V_s}^2 - {V_T}^2} }}}}} \right)}}} \right)}}{{\ln \left( {\frac{{{V_T} + {V_s}{e^{\frac{{2\pi {V_T}}}{{n\sqrt {{V_s}^2 - {V_T}^2} }}}}}}{{{V_s} + {V_T}}}} \right)}}} \right\rceil +\eta + 1
\end{array}
\label{e98992}
\end{equation}
Where $\eta=0$, or $\eta =1$
\\
We denote by ${T_{in}}(n)$ the sum of all the inward advancement times and by ${T_{spiral}}(n)$ the sum of all the spiral traversal times. Therefore the time it takes the swarm to clean the entire evader region is given by,
\begin{equation}
T(n) = {T_{in}}(n) + {T_{spiral}}(n)
\label{e98993}
\end{equation}
Where ${T_{in}}(n)$ is given by,
\begin{equation}
{T_{in}}(n) = \widetilde{{T_{in}}}(n) + {T_{_{in}last}}(n) + \eta {T_{i{n_f}}}(n)
\label{e769}
\end{equation}
Where $\widetilde{{T_{in}}}(n)$ is given by,
\begin{equation}
\begin{array}{l}
\widetilde{T_{in}}(n) = \frac{{2r}}{{{V_s} + {V_T}}} + \frac{{{R_0} - r}}{{{V_s}}} + \frac{{2r\left( {{V_T} + {V_s}{e^{\frac{{2\pi {V_T}}}{{n\sqrt {{V_s}^2 - {V_T}^2} }}}}} \right)}}{{{V_s}\left( {{V_s} + {V_T}} \right)\left( {1 - {e^{\frac{{2\pi {V_T}}}{{n\sqrt {{V_s}^2 - {V_T}^2} }}}}} \right)}}\\
 - \frac{{{{\left( {{V_T} + {V_s}{e^{\frac{{2\pi {V_T}}}{{n\sqrt {{V_s}^2 - {V_T}^2} }}}}} \right)}^{{N_n} - 1}}}}{{{V_s}\left( {{V_s} + {V_T}} \right)\left( {1 - {e^{\frac{{2\pi {V_T}}}{{n\sqrt {{V_s}^2 - {V_T}^2} }}}}} \right)}}\left( {{R_0}\left( {1 - {e^{\frac{{2\pi {V_T}}}{{n\sqrt {{V_s}^2 - {V_T}^2} }}}}} \right) + r\left( {1 + {e^{\frac{{2\pi {V_T}}}{{n\sqrt {{V_s}^2 - {V_T}^2} }}}}} \right)} \right)
\end{array}
\label{e772}
\end{equation}
${T_{_{in}last}}(n)$ is given by,
\begin{equation}
{T_{_{in}last}}(n) = \frac{R_N}{V_s}
\label{e773}
\end{equation}
And ${T_{i{n_f}}}(n)$ is given by,
\begin{equation}
{T_{i{n_f}}}(n) = \frac{{{T_l}{V_T}}}{{{V_s}}}
\label{e768}
\end{equation}
And therefore,
\begin{equation}
{T_{in}}(n) = \widetilde{{T_{in}}}(n) + \frac{{{R_{N}}}}{{{V_s}}} + \frac{{\eta r}}{{{V_s}}}\left( {{e^{\frac{{2\pi {V_T}}}{{n\sqrt {{V_s}^2 - {V_T}^2} }}}} - 1} \right)
\label{e774}
\end{equation}
${T_{spiral}}(n)$ is given by,
\begin{equation}
{T_{spiral}}(n) = \widetilde{{T_{spiral}}}(n) + {T_{last}}(n)+ \eta {T_l}(n)
\label{e775}
\end{equation}
Where $\widetilde{{T_{spiral}}}(n)$ is given by,
\begin{equation}
\begin{array}{l}
\widetilde{{T_{spiral}}}(n) = \frac{{\left( {r - {R_0}} \right)\left( {{V_s} + {V_T}} \right)}}{{{V_T}{V_s}}} - \frac{{2r\left( {{V_T} + {V_s}{e^{\frac{{2\pi {V_T}}}{{n\sqrt {{V_s}^2 - {V_T}^2} }}}}} \right)}}{{{V_T}{V_s}\left( {1 - {e^{\frac{{2\pi {V_T}}}{{n\sqrt {{V_s}^2 - {V_T}^2} }}}}} \right)}} \\ - {\left( {\frac{{{V_T} + {V_s}{e^{\frac{{2\pi {V_T}}}{{n\sqrt {{V_s}^2 - {V_T}^2} }}}}}}{{{V_s} + {V_T}}}} \right)^{{N_n}}}\left( {\frac{{\left( {{V_s} + {V_T}} \right)\left( {{R_0}\left( {{e^{\frac{{2\pi {V_T}}}{{n\sqrt {{V_s}^2 - {V_T}^2} }}}} - 1} \right) - r\left( {{e^{\frac{{2\pi {V_T}}}{{n\sqrt {{V_s}^2 - {V_T}^2} }}}} + 1} \right)} \right)}}{{{V_T}{V_s}\left( {1 - {e^{\frac{{2\pi {V_T}}}{{n\sqrt {{V_s}^2 - {V_T}^2} }}}}} \right)}}} \right) + \frac{{2r\left( {{N_n} - 1} \right)}}{{{V_T}}}
\end{array}
\label{e776}
\end{equation}
${T_{last}}(n)$  is given by,
\begin{equation}
{T_{last}}(n) = \frac{{2\pi r}}{{n{V_s}}}
\label{e778}
\end{equation}
And ${T_l}(n)$ is given by,
\begin{equation}
{T_l}(n) = \frac{{r\left( {{e^{\frac{{2\pi {V_T}}}{{n\sqrt {{V_s}^2 - {V_T}^2} }}}} - 1} \right)}}{{{V_T}}}
\label{e767}
\end{equation}
And therefore ${T_{spiral}}(n)$ is given by,
\begin{equation}
{T_{spiral}}(n) = \widetilde{{T_{spiral}}}(n) + \frac{{2\pi r}}{{{n{V_s}}}} + \eta \frac{{r\left( {{e^{\frac{{2\pi {V_T}}}{{n\sqrt {{V_s}^2 - {V_T}^2} }}}} - 1} \right)}}{{{V_T}}}
\label{e779}
\end{equation}
\end{thm}

\begin{proof}
Let us denote by $\Delta V >0$ the addition to the sweeper's velocity above the critical velocity. The sweeper's velocity is therefore given by, $V_s = V_c + \Delta V$. The expression for the angle that the sweeper travels, denoted as $\theta \left( {{t_\theta }} \right)$, when at the beginning of the cycle the center of the sweeper's sensor is located at a distance of ${R_i} - r$ from the center of the evader region is calculated in (\ref{e15}). Replacing $R_0$ with $R_i$ yields,
\begin{equation}
\theta \left( {{t_\theta }} \right) = \frac{{\sqrt {{V_s}^2 - {V_T}^2} }}{{{V_T}}}\ln \left( {\frac{{{V_T}{t_\theta } + {R_i} - r}}{{{R_i} - r}}} \right)
\label{e46}
\end{equation}
The time it takes the sweeper to travel an angle of $\theta \left( {{t_\theta }} \right)=\frac{2\pi}{n}$ is denoted as ${T_{spiral}}_i$ and is obtained from (\ref{e46}). It is given by,
\begin{equation}
{T_{spiral}}_i = \frac{{\left( {{R_i} - r} \right)\left( {{e^{\frac{{2\pi {V_T}}}{{n\sqrt {{V_s}^2 - {V_T}^2} }}}} - 1} \right)}}{{{V_T}}}
\label{e47}
\end{equation}
Given that an agent moves in a velocity that is greater than the critical velocity for the corresponding scenario, we denote as in the circular sweep process, the distance an agent can advance towards the center of the evader region by ${\delta _i}(\Delta V)$. This will result in a new circular evader region with a radius of ${R_{i + 1}} = {R_i} - {\delta _i}(\Delta V)$.
After completing the proposed spiral sweep the evader region is again circularly shaped, with a smaller radius. A proof for this property is provided in Appendix $H$. We have that,
\begin{equation}
{\delta _i}(\Delta V) = 2r - {V_T}{T_{spiral}}_i
\label{e1015}
\end{equation}
As a function of the number of sweepers and the iteration number, we have that the distance the sweepers can advance inwards after completing an iteration in case the evader wavefront did not continue to expand during the sweepers' inward motion is given by,
\begin{equation}
{\delta _i}(\Delta V) = 2r - \left( {{R_i} - r} \right)\left( {{e^{\frac{{2\pi {V_T}}}{{n\sqrt {{V_s}^2 - {V_T}^2} }}}} - 1} \right)
\label{e1067}
\end{equation}
Where in the term ${\delta _i}(\Delta V)$, $\Delta V$ denotes the increase in the agent velocity relative to the critical velocity, and $i$ denotes the number of sweep iterations the sweepers performed around the evader region, where $i$ starts from sweep number $0$. As in the case of the circular sweep process the time it takes the sweepers to move inwards until their entire sensors are over the evader region depends on the relative velocity between the sweepers inwards entry velocities and the evader region outwards expansion velocity. Therefore the distance a sweeper can advance inwards after completing an iteration is given by,
\begin{equation}
{\delta _{{i_{eff}}}}(\Delta V) = {\delta _i}(\Delta V)\left( {\frac{{{V_s}}}{{{V_s} + {V_T}}}} \right)
\label{e700}
\end{equation}
The new radius of the smaller circular evader region will be therefore given by,
\begin{equation}
{R_{i + 1}} = {R_i} - {\delta _i}(\Delta V)\left( {\frac{{{V_s}}}{{{V_s} + {V_T}}}} \right)
\label{e1097}
\end{equation}
We denote by $\widetilde {R_i}=R_i -r$. Plugging in the value for ${{\delta _i}(\Delta V)}$ into (\ref{e1097}) results in,
\begin{equation}
\widetilde{{R_{i + 1}}} = \widetilde{{R_i}} - \left( {2r -\widetilde{{R_i}}\left( {{e^{\frac{{2\pi {V_T}}}{{n\sqrt {{V_s}^2 - {V_T}^2} }}}} - 1} \right)} \right)\left( {\frac{{{V_s}}}{{{V_s} + {V_T}}}} \right)
\label{e1072}
\end{equation}
Rearranging terms to bring the equation to the difference equation form we observed in the previous section yields,
\begin{equation}
\widetilde{{R_{i + 1}}} = \widetilde{{R_i}}\left( {\frac{{{V_T} + {V_s}{e^{\frac{{2\pi {V_T}}}{{n\sqrt {{V_s}^2 - {V_T}^2} }}}}}}{{{V_s} + {V_T}}}} \right) - \frac{{2r{V_s}}}{{{V_s} + {V_T}}}
\label{e1073}
\end{equation}
Where we denote the coefficients in (\ref{e1073}) as,
\begin{equation}
{c_3} = \frac{{{V_T} + {V_s}{e^{\frac{{2\pi {V_T}}}{{n\sqrt {{V_s}^2 - {V_T}^2} }}}}}}{{{V_s} + {V_T}}}
\label{e1074}
\end{equation}
\begin{equation}
{c_1} =  - \frac{{2r{V_s}}}{{{V_s} + {V_T}}}
\label{e1075}
\end{equation}
This yields the following difference equation,
\begin{equation}
\widetilde{{R_{i + 1}}} = {c_3}\widetilde{{R_i}} +{c_1}
\label{e345}
\end{equation}
Since the structure of the difference equation for the radius of the evader region is the same as those in the circular case described in the previous section, the number of iterations it takes the sweepers to reduce the evader region to a circle of radius $\widehat{{R_N}}=2r$, is given by,
\begin{equation}
{N_n} = \left\lceil {\frac{{\ln \left({\frac{{\widehat{{R_N}} - \frac{{{c_1}}}{{1 - {c_3}}}}}{{{R_0} - \frac{{{c_1}}}{{1 - {c_3}}}}}} \right)}}{{\ln {c_3}}}} \right\rceil
\label{e1029}
\end{equation}
${R_N}$ is the actual radius of the circular evader region whose radius is smaller or equal to $2r$ and is calculated by similar steps as ${R_{N-2}}$ is calculated in Appendix $B$. The precise calculation of ${R_N}$ is important for the end game of the cleaning process as is discussed later in this section. The last cycle takes place when the evader region is a circle of radius $\widehat{{R_N}} = 2r$, or $\widetilde{{R_N}} = r$. Thus substitution of coefficients in (\ref{e1029}) yields that after ${N_n}$ iterations the evader region is circularly shaped with a radius that is less than or equal to $\frac{2r}{n}$. ${N_n}$ is given by,
\begin{equation}
{N_n} = \left\lceil {\frac{{\ln \left( {\frac{{r\left( {3 - {e^{\frac{{2\pi {V_T}}}{{n\sqrt {{V_s}^2 - {V_T}^2} }}}}} \right)}}{{{R_0}\left( {1 - {e^{\frac{{2\pi {V_T}}}{{n\sqrt {{V_s}^2 - {V_T}^2} }}}}} \right) + r\left( {1 + {e^{\frac{{2\pi {V_T}}}{{n\sqrt {{V_s}^2 - {V_T}^2} }}}}} \right)}}} \right)}}{{\ln \left( {\frac{{{V_T} + {V_s}{e^{\frac{{2\pi {V_T}}}{{n\sqrt {{V_s}^2 - {V_T}^2} }}}}}}{{{V_s} + {V_T}}}} \right)}}} \right\rceil
\label{e1076}
\end{equation}
The inwards advancement time depends on the iteration number. It is denoted by $T_{i{n_i}}$ and its expression is given by,
\begin{equation}
{T_{i{n_i}}} = \frac{{{\delta _{{i_{eff}}}}(\Delta V)}}{{{V_s}}} = \frac{{2r - {R_i}\left( {{e^{\frac{{2\pi {V_T}}}{{n\sqrt {{V_s}^2 - {V_T}^2} }}}} - 1} \right)}}{{{V_s} + {V_T}}}
\label{e1077}
\end{equation}
We denote the total advancement time until the evader region is reduced to a circle of radius that is less than or equal to $2r$ as $\widetilde{{T_{in}}}(n)$. It is given by,
\begin{equation}
\widetilde{{T_{in}}}(n) = \sum\limits_{i = 0}^{{N_n}-2} {{T_{i{n_i}}}}
\label{e1098}
\end{equation}
Since during the inward advancements only the tip of the sensor, that has zero width, is inserted into the evader region it does not detect evaders until it completes its inward advance and starts sweeping again. After the sweepers complete their advance into the evader region their sensor footprint over the evader region is equal to $2r$. The total search time until the evader region is reduced to a circle with a radius that is less than or equal to $2r$  is given by the sum of the total spiral sections times combined with the times of the total inward advancements. Namely,
\begin{equation}
\widetilde{T}(n) = \widetilde{{T_{in}}}(n) + \widetilde{{T_{spiral}}}(n)
\label{e1083}
\end{equation}
Using the developed term for $T_{i{n_i}}$ the total inward advancement times until the evader region is reduced to a circle with a radius that is less than or equal to $2r$ are computed by,
\begin{equation}
\widetilde{{T_{in}}}(n) = \sum\limits_{i = 0}^{{N_n} - 2} {{T_{i{n_i}}}}  = \sum\limits_{i = 0}^{{N_n} - 2} {\frac{{2r - \widetilde{{R_i}}\left( {{e^{\frac{{2\pi {V_T}}}{{n\sqrt {{V_s}^2 - {V_T}^2} }}}} - 1} \right)}}{{{V_s} + {V_T}}}}
\label{e702}
\end{equation}
The full derivation of $\widetilde{{T_{in}}}(n) = \sum\limits_{i = 0}^{{N_n} - 2} {{T_{i{n_i}}}}$ is given in Appendix $G$. We therefore have that,
\begin{equation}
\begin{array}{l}
\widetilde{T_{in}}(n) = \frac{{2r}}{{{V_s} + {V_T}}} + \frac{{{R_0} - r}}{{{V_s}}} + \frac{{2r\left( {{V_T} + {V_s}{e^{\frac{{2\pi {V_T}}}{{n\sqrt {{V_s}^2 - {V_T}^2} }}}}} \right)}}{{{V_s}\left( {{V_s} + {V_T}} \right)\left( {1 - {e^{\frac{{2\pi {V_T}}}{{n\sqrt {{V_s}^2 - {V_T}^2} }}}}} \right)}}\\
 - \frac{{{{\left( {{V_T} + {V_s}{e^{\frac{{2\pi {V_T}}}{{n\sqrt {{V_s}^2 - {V_T}^2} }}}}} \right)}^{{N_n} - 1}}}}{{{V_s}\left( {{V_s} + {V_T}} \right)\left( {1 - {e^{\frac{{2\pi {V_T}}}{{n\sqrt {{V_s}^2 - {V_T}^2} }}}}} \right)}}\left( {{R_0}\left( {1 - {e^{\frac{{2\pi {V_T}}}{{n\sqrt {{V_s}^2 - {V_T}^2} }}}}} \right) + r\left( {1 + {e^{\frac{{2\pi {V_T}}}{{n\sqrt {{V_s}^2 - {V_T}^2} }}}}} \right)} \right)
\end{array}
\label{e1084}
\end{equation}
In the last inward advancement towards the center of the evader region the sweepers advance inwards and place the lower tips of their sensors at the center of the evader region. The time it takes the sweepers to complete this inwards motion is given by,
\begin{equation}
{T_{_{in}last}}(n) = \frac{{{R_{{N}}} }}{{{V_s}}}
\label{e705}
\end{equation}
${R_{{N}}}$ is calculated by similar steps as the calculation in Appendix $B$. Recalling that $\widetilde{{R_{{N}}}} = {R_{{N}}} -r $ we have that,
\begin{equation}
\widetilde{{R_{{N}}}} = \frac{{{c_1}}}{{1 - {c_3}}} + {c_3}^{{N_n}}\left( {\widetilde{{R_0}} - \frac{{{c_1}}}{{1 - {c_3}}}} \right)
\label{e781}
\end{equation}
Substituting the coefficients in (\ref{e781}) yields,
\begin{equation}
{R_N} =  - \frac{{2r}}{{1 - {e^{\frac{{2\pi {V_T}}}{{n\sqrt {{V_s}^2 - {V_T}^2} }}}}}}
 + {c_3}^{{N_n}}\left( {\frac{{{R_0}\left( {1 - {e^{\frac{{2\pi {V_T}}}{{n\sqrt {{V_s}^2 - {V_T}^2} }}}}} \right) + r\left( {1 + {e^{\frac{{2\pi {V_T}}}{{n\sqrt {{V_s}^2 - {V_T}^2} }}}}} \right)}}{{1 - {e^{\frac{{2\pi {V_T}}}{{n\sqrt {{V_s}^2 - {V_T}^2} }}}}}}} \right)
\label{e782}
\end{equation}
Substituting the expression for ${R_{{N}}}$ in ${T_{_{in}last}}(n)$ given in (\ref{e705}) yields,
\begin{equation}
{T_{_{in}last}}(n) = - \frac{{2r}}{{{V_s}\left( {1 - {e^{\frac{{2\pi {V_T}}}{{n\sqrt {{V_s}^2 - {V_T}^2} }}}}} \right)}}
 + {c_3}^{{N_n}}\left( {\frac{{{R_0}\left( {1 - {e^{\frac{{2\pi {V_T}}}{{n\sqrt {{V_s}^2 - {V_T}^2} }}}}} \right) + r\left( {1 + {e^{\frac{{2\pi {V_T}}}{{n\sqrt {{V_s}^2 - {V_T}^2} }}}}} \right)}}{{{V_s}\left( {1 - {e^{\frac{{2\pi {V_T}}}{{n\sqrt {{V_s}^2 - {V_T}^2} }}}}} \right)}}} \right)
\label{e727}
\end{equation}
Since the time to sweep around radius $\widetilde{R_i}$ is obtained by multiplying $\widetilde{R_i}$ by $\frac{{{e^{\frac{{2\pi {V_T}}}{{n\sqrt {{V_s}^2 - {V_T}^2} }}}} - 1}}{V_T}$, when multiplying (\ref{e1073}) by  $\frac{{{e^{\frac{{2\pi {V_T}}}{{n\sqrt {{V_s}^2 - {V_T}^2} }}}} - 1}}{V_T}$ we can construct a difference equation for the sweep times. This difference equation is given by,
\begin{equation}
{T_{i + 1}} = {c_3}{T_i} + {c_4}
\label{e37}
\end{equation}
Where in this context a cycle is defined as a traversal of an angle of $\frac{2\pi}{n}$ by the sweeper. The coefficient $c_4$ is given by,
\begin{equation}
{c_4} = \frac{{ - 2r{V_s}\left( {{e^{\frac{{2\pi {V_T}}}{{n\sqrt {{V_s}^2 - {V_T}^2} }}}} - 1} \right)}}{{({V_s}+{V_T}){V_T}}}
\label{e1062}
\end{equation}
The total time of the spiral sweeps until the evader region is reduced to a circle with a radius that is equal to or smaller than $2r$ follows the derivation in Appendix $C$ and is given by,
\begin{equation}
\widetilde{{T_{spiral}}}(n) = \frac{{{T_0} - {c_3}{T_{N_n - 1}} + \left( {{N_n} - 1} \right){c_4}}}{{1 - {c_3}}}
\label{e100}
\end{equation}
Where the time of the first sweep is given by,
\begin{equation}
{T_0} = \frac{{\left( {{R_0} - r} \right)\left( {{e^{\frac{{2\pi {V_T}}}{{n\sqrt {{V_s}^2 - {V_T}^2} }}}} - 1} \right)}}{{{V_T}}}
\label{e44}
\end{equation}
The time it takes to sweep the last cycle of the search process before the evader region is reduced to a circle with a radius that is less than or equal to $2r$ is computed in Appendix $D$, and is given by,
\begin{equation}
{T_{N_n - 1}} = \frac{{{c_4}}}{{1 - {c_3}}} +  {{c_3}^{{N_n} - 1}} \left( {{T_0} - \frac{{{c_4}}}{{1 - {c_3}}}} \right)
\label{e1089}
\end{equation}
Plugging in the appropriate coefficients yields,
\begin{equation}
{T_{N - 1}} = \frac{{2r}}{{{V_T}}}
 + \frac{{{c_3}^{{N_n} - 1}}}{{{V_T}}}\left( {{R_0}\left( {{e^{\frac{{2\pi {V_T}}}{{n\sqrt {{V_s}^2 - {V_T}^2} }}}} - 1} \right) - r\left( {1 + {e^{\frac{{2\pi {V_T}}}{{n\sqrt {{V_s}^2 - {V_T}^2} }}}}} \right)} \right)
\label{e729}
\end{equation}
Substituting the derived coefficients into (\ref{e100}) yields,
\begin{equation}
\begin{array}{l}
\widetilde{{T_{spiral}}}(n) = \frac{{\left( {r - {R_0}} \right)\left( {{V_s} + {V_T}} \right)}}{{{V_T}{V_s}}} - \frac{{2r\left( {{V_T} + {V_s}{e^{\frac{{2\pi {V_T}}}{{n\sqrt {{V_s}^2 - {V_T}^2} }}}}} \right)}}{{{V_T}{V_s}\left( {1 - {e^{\frac{{2\pi {V_T}}}{{n\sqrt {{V_s}^2 - {V_T}^2} }}}}} \right)}} \\ - {\left( {\frac{{{V_T} + {V_s}{e^{\frac{{2\pi {V_T}}}{{n\sqrt {{V_s}^2 - {V_T}^2} }}}}}}{{{V_s} + {V_T}}}} \right)^{{N_n}}}\left( {\frac{{\left( {{V_s} + {V_T}} \right)\left( {{R_0}\left( {{e^{\frac{{2\pi {V_T}}}{{n\sqrt {{V_s}^2 - {V_T}^2} }}}} - 1} \right) - r\left( {{e^{\frac{{2\pi {V_T}}}{{n\sqrt {{V_s}^2 - {V_T}^2} }}}} + 1} \right)} \right)}}{{{V_T}{V_s}\left( {1 - {e^{\frac{{2\pi {V_T}}}{{n\sqrt {{V_s}^2 - {V_T}^2} }}}}} \right)}}} \right) + \frac{{2r\left( {{N_n} - 1} \right)}}{{{V_T}}}
\end{array}
\label{e1086}
\end{equation}
After completing sweep number ${{N_n} - 1}$ the sweepers advance toward the center of the evader region until the lower tips of their sensors are located at the center of the evader region. Following this advance the sweepers need to perform a circular sweep of radius $r$ around the center of the evader region and complete the cleaning of the evader region. The sweepers can complete this last circular sweep only if their velocities are high enough so that during the circular motion no evader escapes the sweepers. Since the critical velocity for a spiral sweep is lower than the critical velocity for a circular sweep the sweepers need to perform the last circular after spiral sweep number ${{N_n} - 1}$ only if their velocities satisfy the following inequality,
\begin{equation}
2r \ge {V_T}{T_{last}} + {V_T}{T_{_{in}last}} + {R_{N}}
\label{e703}
\end{equation}
Satisfying (\ref{e703}) means that no evader escapes the sweepers. Before the last sweep the evader region is reduced to a circle of radius ${R_{N}}$ that satisfies,
\begin{equation}
0 < {R_{N}} \le 2r
\label{e704}
\end{equation}
An alternative way to represent ${R_{{N}}}$ is,
\begin{equation}
{R_{{N}}} = r\left( {2 - \varepsilon } \right)
\label{e618}
\end{equation}
Therefore $\varepsilon$ can be written as,
\begin{equation}
\varepsilon  = \frac{{2r - {R_N}}}{r}
\label{e621}
\end{equation}
Where,
\begin{equation}
0 \le \varepsilon  < 2
\label{e711}
\end{equation}
The last circular sweep occurs after the sweepers advance towards the center of the evader region and place the lower tips of their sensors at the center of the evader region. The last sweep is therefore a circular sweep of an angle of $\frac{2\pi}{n}$  around a circle with a radius of $r$ that is centered at the center of the evader region. The time it takes the sweepers to complete it is given by,
\begin{equation}
{T_{last}}(n) = \frac{{2\pi r}}{{n{V_s}}}
\label{e619}
\end{equation}
Therefore in order to perform the last circular sweep directly after spiral sweep number ${{N_n} - 1}$, the inequality in (\ref{e703}), that can be written as,
\begin{equation}
\varepsilon  \ge \frac{{2\pi {V_T} + n{V_T}\left( {2 - \varepsilon } \right)}}{{n{V_s}}}
\label{e708}
\end{equation}
has to hold. Therefore in order to perform the last circular sweep directly after spiral sweep number ${{N_n} - 1}$, $V_s$ has to satisfy that,
\begin{equation}
{V_s} \ge \frac{{2\pi {V_T} + n{V_T}\left( {2 - \varepsilon } \right)}}{{n\varepsilon }}
\label{e709}
\end{equation}
Substituting $\varepsilon$ in (\ref{e709}) with the expression for $\varepsilon$ from (\ref{e621}) yields that in order to perform the last circular sweep directly after spiral sweep number ${{N_n} - 1}$, $V_s$ has to satisfy that,
\begin{equation}
{V_s} \ge \frac{{2\pi r{V_T} + n{V_T}{R_N}}}{{n\left( {2r - {R_N}} \right)}}
\label{e622}
\end{equation}
Rearranging terms and denoting the smallest possible $\varepsilon$ that satisfies (\ref{e709}) as ${\varepsilon _c}$, yields that,
\begin{equation}
{\varepsilon _c} \ge \frac{{2{V_T}\left( {\pi  + n} \right)}}{{n\left( {{V_s} + {V_T}} \right)}}
\label{e710}
\end{equation}
Therefore if the radius of the circular evader region after sweep number ${{N_n} - 1}$ satisfies,
\begin{equation}
{R_N} \ge r\left( {2 - {\varepsilon _c}} \right)
\label{e712}
\end{equation}
or,
\begin{equation}
{R_N} \ge \frac{{2rn{V_s} - 2{V_T}\pi r}}{{{n^2}\left( {{V_s} + {V_T}} \right)}}
\label{e713}
\end{equation}
Then the sweepers velocity is not sufficient to guarantee escape from the evader region. The demand that ${{R_{{N}}}}$ will satisfy the inequality in (\ref{e713}) is equivalent to the demand that if $V_s$ is not high enough and does not satisfy the inequality in (\ref{e622}) then the sweepers velocity is not sufficient to guarantee escape. Therefore the sweepers have to perform another spiral sweep, that starts when the lower tips of their sensors are located at the center of the evader region. This spiral sweep starts when the center of each sweeper is at a distance of $r$ from the center of the region and the time it takes to complete it is denoted by ${T_l}(n)$. It is given by,
\begin{equation}
{T_l}(n) = \frac{{r\left( {{e^{\frac{{2\pi {V_T}}}{{n\sqrt {{V_s}^2 - {V_T}^2} }}}} - 1} \right)}}{{{V_T}}}
\label{e715}
\end{equation}
We therefore introduce a characteristic function named $\eta$ that takes the values of $1$ and $0$. If the additional spiral sweep needs to be performed $\eta =1$ and therefore ${T_l}(n)$ is added to the sweep time and if no additional spiral sweep is needed $\eta =0$. Therefore the general term for ${T_{spiral}}(n)$ is given by,
\begin{equation}
{T_{spiral}}(n) = \widetilde{{T_{spiral}}}(n) + {T_{last}}(n) + \eta {T_l}(n)
\label{e719}
\end{equation}
${T_{in}}(n)$ is given by the sum,
\begin{equation}
{T_{in}}(n) = \widetilde{{T_{in}}}(n) + {T_{_{in}last}}(n) + \eta {T_{i{n_f}}}(n)
\label{e717}
\end{equation}
Let us denote by ${T_{i{n_f}}}(n)$ the time of the inward advancement of the sweepers that corresponds to the spread of possible evaders that originated at the center of the evader region at the beginning of the last spiral sweep and had time of ${T_l}(n)$ to spread from the center at a velocity of $V_T$. Therefore ${T_{i{n_f}}}(n)$ is given by,
\begin{equation}
{T_{i{n_f}}}(n) = \frac{{{T_l}(n){V_T}}}{{{V_s}}}
\label{e716}
\end{equation}
And the total times of inward advancements is given by,
\begin{equation}
{T_{in}}(n) = \widetilde{{T_{in}}}(n) + \frac{{{R_{N}}}}{{{V_s}}} + \frac{{\eta r\left( {{e^{\frac{{2\pi {V_T}}}{{n\sqrt {{V_s}^2 - {V_T}^2} }}}} - 1} \right)}}{{{V_s}}}
\label{e718}
\end{equation}
Substituting the terms in (\ref{e719}) yields,
\begin{equation}
{T_{spiral}}(n) = \widetilde{{T_{spiral}}}(n) + \frac{{2\pi r}}{{{n}{V_s}}} + \frac{{\eta r\left( {{e^{\frac{{2\pi {V_T}}}{{n\sqrt {{V_s}^2 - {V_T}^2} }}}} - 1} \right)}}{{{V_T}}}
\label{e720}
\end{equation}
\end{proof}

\begin{figure}[ht]
\noindent \centering{}\includegraphics[width=3.4in,height=2.8in]{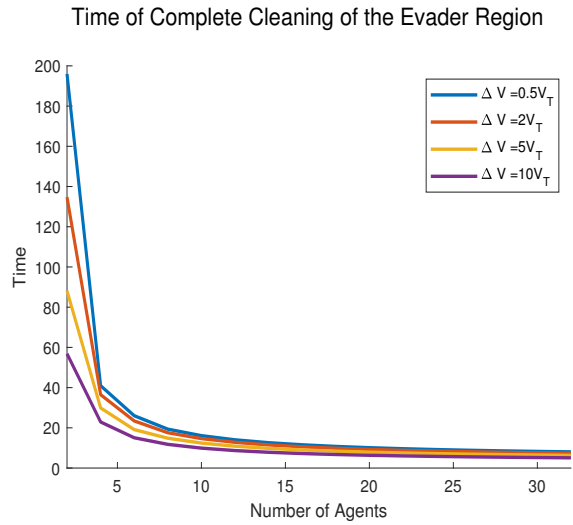} \caption{Time of complete cleaning of the evader region. In this figure we simulated the sweep processes for an even number of agents, ranging from $2$ to $32$ agents that employ the multi-agent spiral sweep process. We show the results obtained for different values of velocities above the spiral critical velocity, i.e. different choices for $\Delta V$. The parameters values chosen for this plot are $r=10$, $V_T = 1$ and $R_0 = 100$.}
\label{Fig12Label}
\end{figure}

\begin{figure}[ht]
\noindent \centering{}\includegraphics[width=3.4in,height=2.8in]{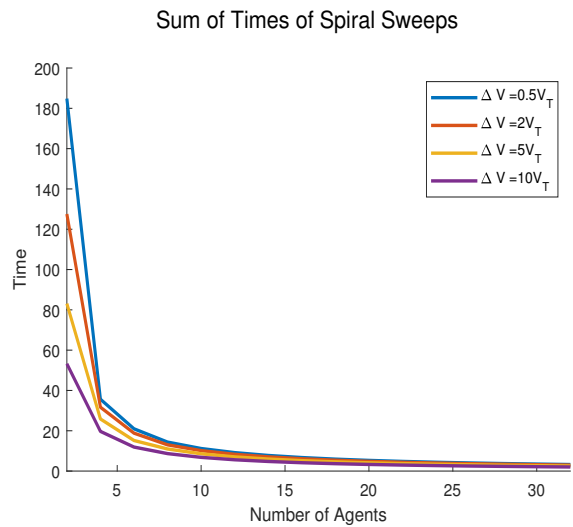} \caption{Sum of the spiral sweep times of the search until complete cleaning of the evader region. In this figure we simulated the sweep processes for an even number of agents, ranging from $2$ to $32$ agents that employ the multi-agent spiral sweep process. We show the results obtained for different values of velocities above the spiral critical velocity, i.e. different choices for $\Delta V$. The parameters values chosen for this plot are $r=10$, $V_T = 1$ and $R_0 = 100$.}
\label{Fig13Label}
\end{figure}

\begin{figure}[ht]
\noindent \centering{}\includegraphics[width=3.4in,height=2.8in]{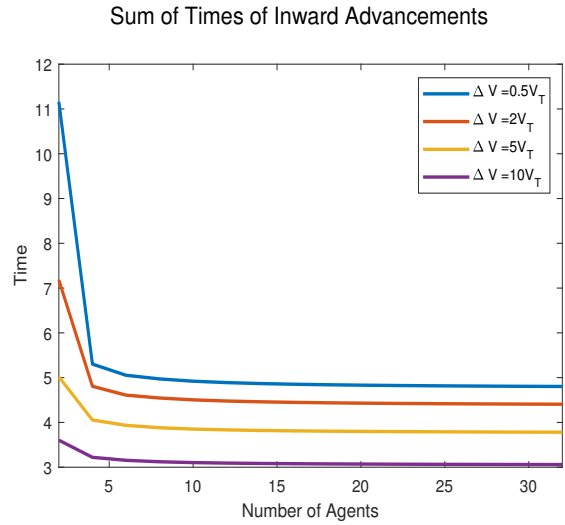} \caption{Sum of the inward advancement times until complete cleaning of the evader region. In this figure we simulated the sweep processes for an even number of agents, ranging from $2$ to $32$ agents that employ the multi-agent spiral sweep process. We show the results obtained for different values of velocities above the spiral critical velocity, i.e. different choices for $\Delta V$. The parameters values chosen for this plot are $r=10$, $V_T = 1$ and $R_0 = 100$.}
\label{Fig14Label}
\end{figure}

\begin{figure}[ht]
\noindent \centering{}\includegraphics[width=3.4in,height=2.8in]{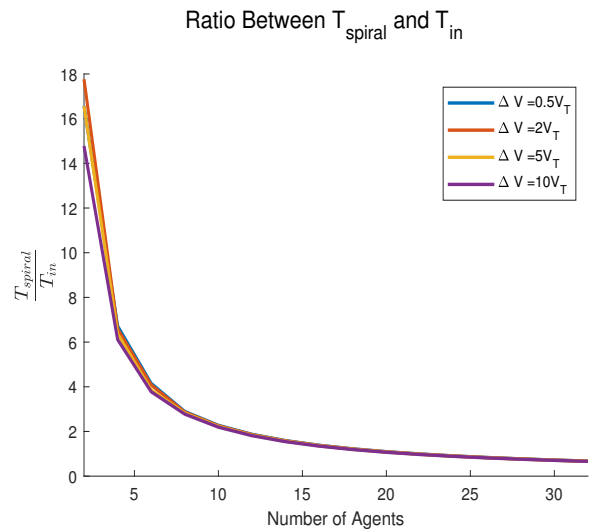} \caption{Ratio between the spiral sweep times of the search and the inward advancement times until complete cleaning of the evader region. In this figure we simulated the sweep processes for an even number of agents, ranging from $2$ to $32$ agents that employ the multi-agent spiral sweep process. We show the results obtained for different values of velocities above the spiral critical velocity, i.e. different choices for $\Delta V$. The parameters values chosen for this plot are $r=10$, $V_T = 1$ and $R_0 = 100$.}
\label{Fig15Label}
\end{figure}

\begin{figure}[ht]
\noindent \centering{}\includegraphics[width=3.4in,height=2.8in]{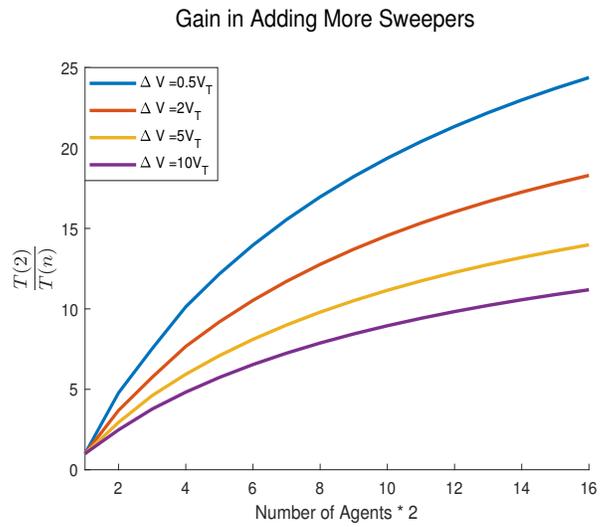} \caption{Gain in cleaning time obtained by adding more sweepers. In this figure we simulated the sweep processes for an even number of agents, denoted by $n$, ranging from $2$ to $32$ agents that employ the multi-agent spiral sweep process. In each of the curves, every point is obtained by the ratio between the sweep times of a $2$ agent swarm and an $n$ agent sweeper swarm. We show the results obtained for different values of velocities above the spiral critical velocity, i.e. different choices for $\Delta V$. The parameters values chosen for this plot are $r=10$, $V_T = 1$ and $R_0 = 100$.}
\label{Fig16Label}
\end{figure}
\section{Comparative Analysis of the Proposed Search Strategies}
The purpose of this section is to compare between the obtained results for the circular and spiral sweep processes that were developed in the previous sections. In order to compare the total sweep times of sweeper swarms that can perform both the circular and spiral sweep processes the number of sweepers and sweeper velocity must be the same in each of the tested spiral and circular swarms in order to make a fair comparison. The critical velocity that is required for sweepers that perform the circular sweep process is higher than the minimal critical velocity that is required for sweepers that perform the spiral sweep process. Therefore we show the results obtained for different values of velocities above the critical velocity that correspond to the circular sweep process. This is done in order to compare the times it takes a multi agent sweeper swarm to clean the entire evader region when moving at the same velocity as the sweepers that employ the corresponding circular sweep. This implies that there are entire regions of operation where an evader area with a given radius could be cleaned by a sweeper swarm that performs the spiral sweep process but cannot be cleaned by a sweeper swarm that performs the circular sweep process. In figures $17$- $18$ we plot the performance of the spiral sweep process with velocities that are above the circular critical velocity. This means that the values of $ \Delta V$ that are mentioned in the plots correspond to sweeper velocities that are almost twice the spiral critical velocities. From comparing these results to the performance of the circular sweep process under the same velocities and number of sweeper agents we clearly see that spiral methods are much more efficient than the circular sweep methods since the times it takes the sweepers to clean the evader region are much lower. Fig. $19$. compares the cleaning times of circular sweeping swarms and spiral sweeping swarms. The results are shown on the same plot and with the same agent velocities for both the circular and spiral multi-agent sweep processes. We note the gap in terms of the times it takes a multi-agent swarm that employs the spiral sweep process to complete its mission compared to the times it takes the same swarm to complete the search when it employs the circular sweep process. Therefore in terms of completion times of the search, the spiral sweep process is favourable. This result is independent of the number of the sweepers that perform the search or the velocity in which they move.

\begin{figure*}[ht]
\begin{minipage}[b]{\linewidth}
\begin{minipage}[b]{0.5\linewidth}
\centering
\includegraphics [width=3.4in,height=2.8in]{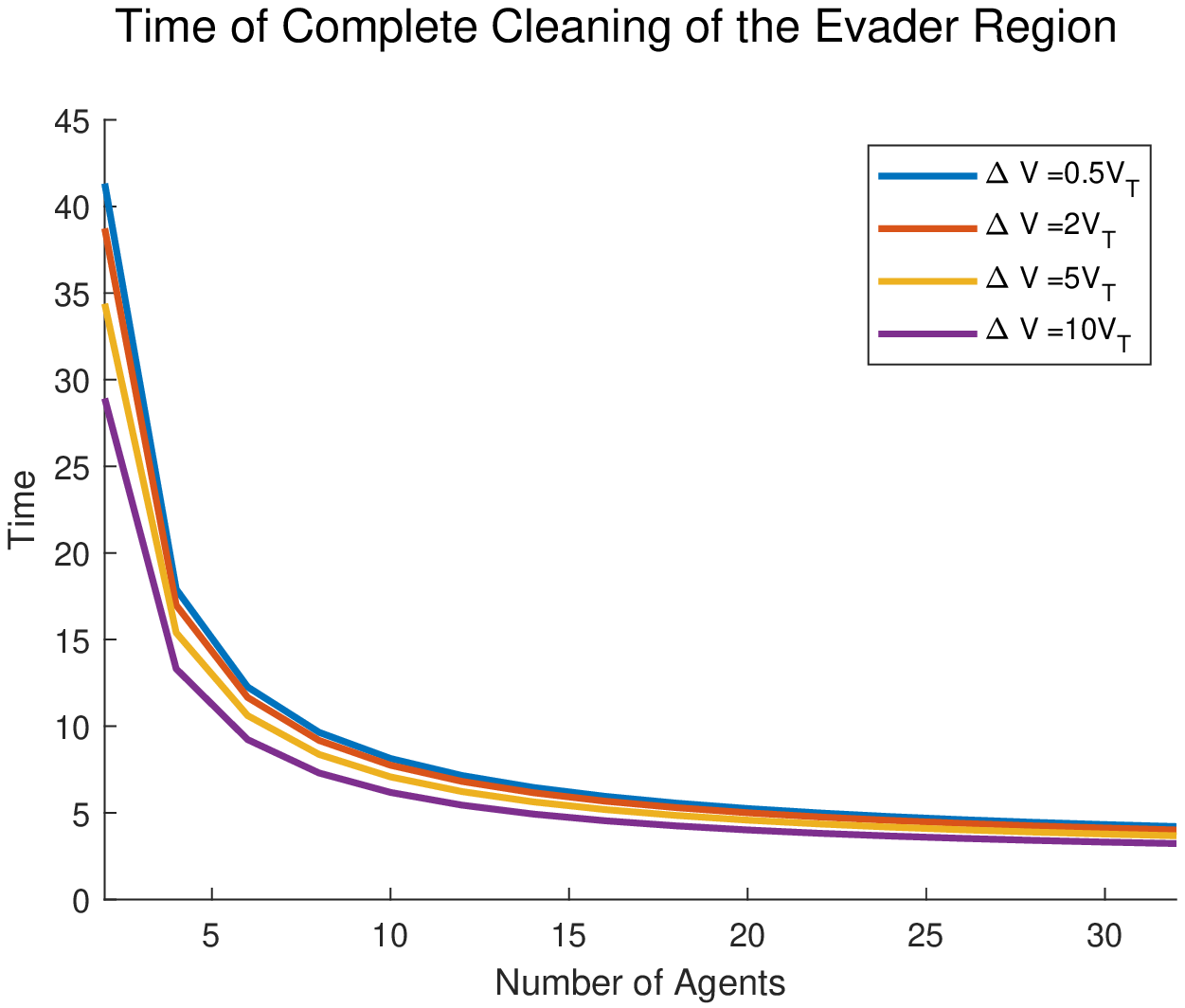}

\end{minipage}
\begin{minipage}[b]{0.5\linewidth}
\centering
\includegraphics [width=3.4in,height=2.8in]{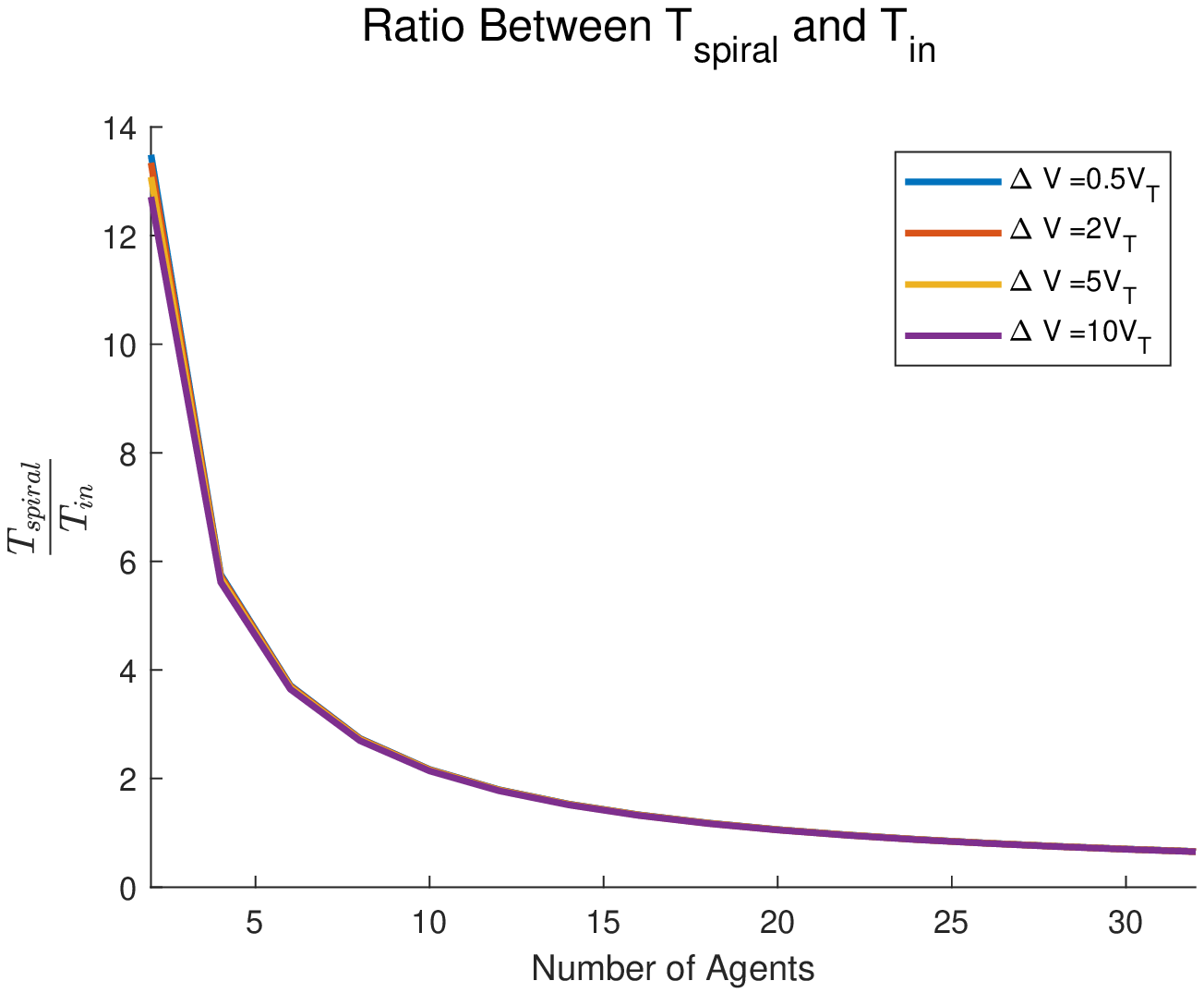}
\end{minipage}
\end{minipage}

\begin{minipage}[b]{\linewidth}
\begin{minipage}[b]{0.5\linewidth}
\centering
\includegraphics [width=3.4in,height=2.8in]{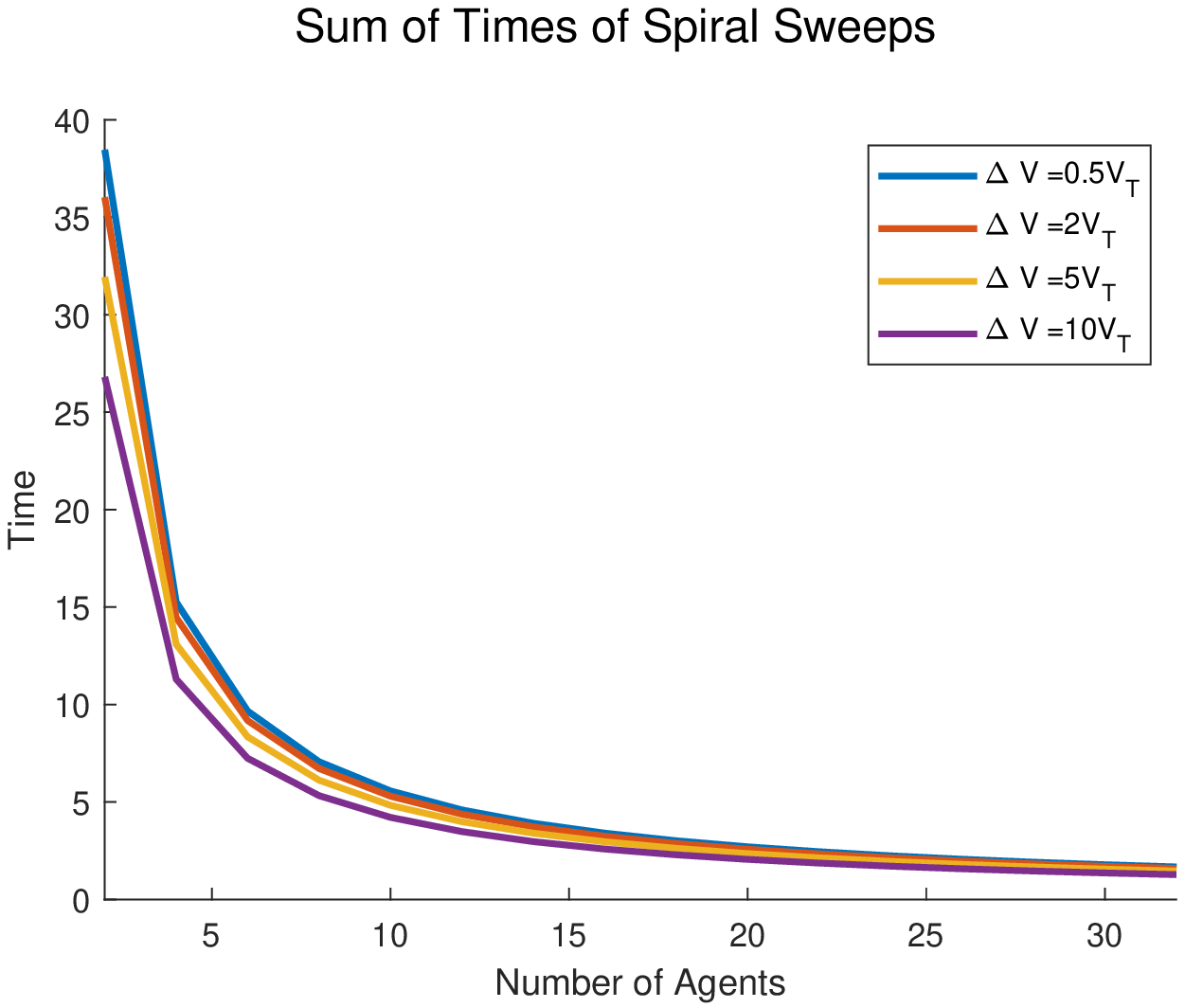} 
\end{minipage}
\begin{minipage}[b]{0.5\linewidth}
\centering
\includegraphics [width=3.4in,height=2.8in]{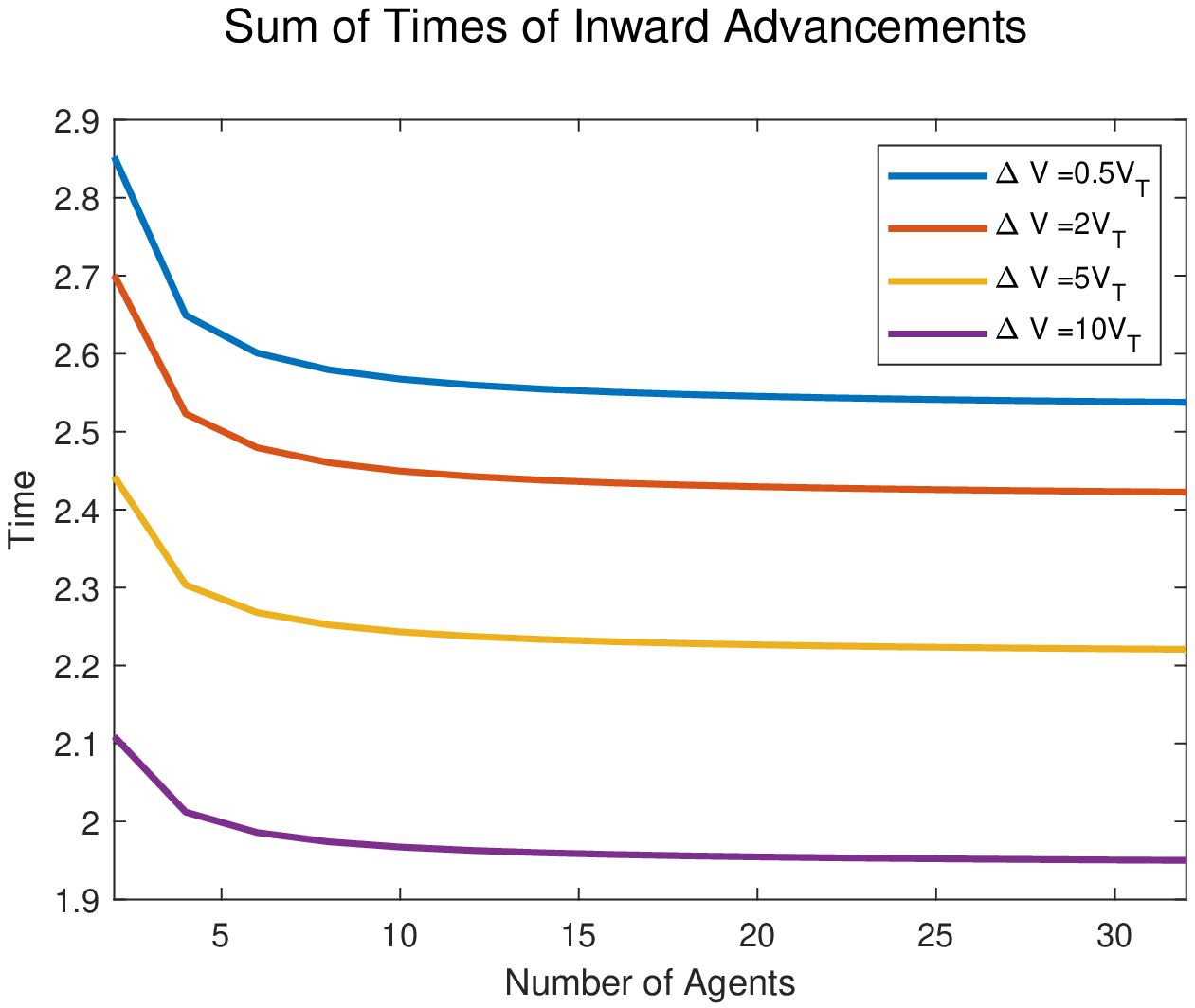}
\end{minipage}
\end{minipage}
\caption{The top left figure shows the time of complete cleaning of the evader region. The top right figure shows the ratio between the spiral sweep times of the search and the inward advancement times. The bottom left figure shows the total spiral sweep times of the search. The bottom right figure shows the total inward advancement times. In all figures we simulated the sweep processes for an even number of agents, ranging from $2$ to $32$ agents that employ the multi-agent spiral sweep process until complete cleaning of the evader region. The parameters values chosen for these plots are $r=10$, $V_T = 1$ and $R_0 = 100$. The values of $\Delta V$ are above the critical velocity of the circular sweep process. }
\end{figure*}

\begin{figure}[ht]
\noindent \centering{}\includegraphics[width=3.4in,height=2.8in]{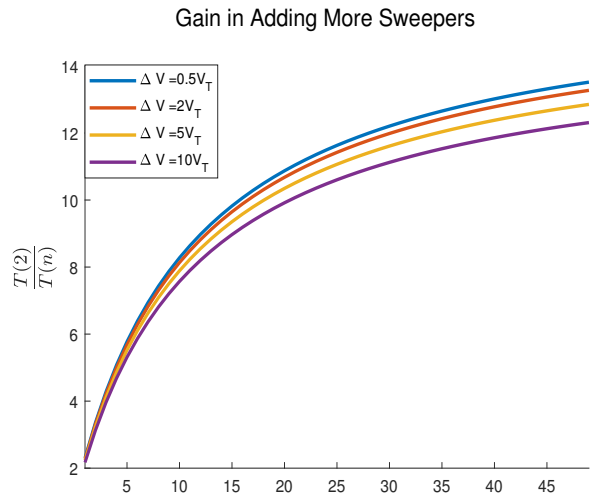} \caption{ Gain in cleaning time obtained by adding more sweepers. In this figure we simulated the sweep processes for an even number of agents, denoted by $n$, ranging from $2$ to $32$ agents that employ the multi-agent spiral sweep process where the inward advancements are continuous. In each of the curves, every point is obtained by the ratio between the sweep times of a $2$ agent swarm and an $n$ agent sweeper swarm. The values of $\Delta V$ are above the critical velocity of the circular sweep process.The parameters values chosen for this plot are $r=10$, $V_T = 1$ and $R_0 = 100$.}
\label{Fig18Label}
\end{figure}

\begin{figure}[ht]
\noindent \centering{}\includegraphics[width=3.4in,height=2.8in]{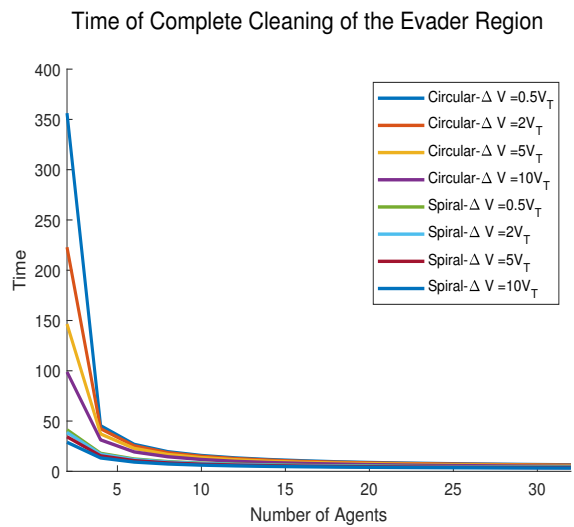} \caption{Total search times until complete cleaning of the evader region for the circular and spiral sweep processes in the case where the inward advancements do not occur in zero time. In this figure we simulated the sweep processes for an even number of agents, ranging from $2$ to $32$ agents, that employ the multi-agent circular and spiral sweep processes. We show the results obtained for different values of velocities above the circular critical velocity. The parameters values chosen for this plot are $r=10$, $V_T = 1$ and $R_0 = 100$.}
\label{Fig19Label}
\end{figure}

\section{Conclusions and Future Research Directions}
This research considers a scenario in which a multi-agent swarm of agents performs a search of an area containing smart mobile evaders that are to be detected. There can be many evaders that we wish do detect in this area, and we consider the domain to be continuous, meaning that an evader can be located at any point in an initial circular region of radius $R_0$. The sweepers are designed in a way that will require a minimal amount of memory in order to complete the required task since the sweeping protocol is predetermined and deterministic. The information the agents perceive only comes from their own sensors, and evaders that intersect a sweeper's linear field of view are detected. Every agent in the swarm has a line shaped sensor of length $2r$. All sweepers move with a speed of $V_s$ that is measured at the center of the linear sensor. The evaders move at a maximal speed of $V_T$. The sweepers objective is to ``clean" or to detect all evaders that can move freely in all directions from their initial locations in the circular region of radius $R_0$. The search time depends on the type of sweeping movement the sweeper swarm employs. The detection of evaders is done using a sweeping method around the region. The desired result is that after each sweep the radius of the circle that bounds the evader region will decrease by a value that is bounded away from zero. This will guarantee a complete cleaning of the evader region, by shrinking the possible area in which evaders can reside to zero, in finite time. We provide an analysis of the proposed sweep processes performance in terms of completion times of the search processes, that is the time at which all potential evaders that resided in the initial evader region were located. We also provide a global balance of covered areas argument to obtain a lower bound on a searcher velocity that is independent of the search process. Based on this point of view, a critical velocity that depends on a novel proposed circular search pattern for a multi agent swarm is derived. This search pattern ensures the satisfaction of the confinement task. This velocity is then compared to the lower bound on the critical velocity that is derived from a maximal swept area versus a minimal danger zone expansion area balance equation. Additionally, we show that the minimal agent velocity that ensures the satisfaction of the confinement task for the proposed circular search process is equal to twice the lower bound and hence is not optimal. An additional critical velocity that depends on a spiral search pattern for a multi agent swarm is derived ensuring the satisfaction of the confinement task as well. We show that the developed spiral critical velocity approaches the theoretical optimal critical velocity. Afterwards, we provide a comparison between the different search methods in terms of completion times of the sweep processes. When comparing the different search processes we compared both the total cleaning times as well as the minimal searcher velocity that are required for a successful search. A possible extension to this work is to replace the assumption made that whenever the evader is in the field of view of the searcher's sensor the evader is detected with probability $1$. A more realistic model would be a model which is based on the statistical properties of the actual detector so that the detection probability is maximized while minimizing the probability of false alarms. Another possible extension could be to introduce a tuneable robustness parameter that will compensate for an agent inability to move in a constant speed throughout the entire search. This robustness parameter will allow a margin between the evader region boundaries and the tip of the agent's sensor that is the farthest away from the center of the evader region. This margin will come at the expense of permitting the agents to advance a lesser distance into the evader region after each iteration. Such a method can compensate for the simplified model where the agents change their directions of travel in zero time for the planar sweep scenario and provide a more realistic model that relies on the conducted research. The assumption that the sweepers change their directions of travel in zero time can be also avoided if the sweepers sweep at different heights above the evader region as discussed.

\appendices
\section{}
The number of sweep iterations that are required to reduce the evader region to be bounded by a circle with a radius that is less or equal to $R_N$ is calculated in the following manner. We have that,
\begin{equation}
{R_{i + 1}} = {c_3}{R_i} + {c_1}
\label{e930}
\end{equation}
Therefore,
\begin{equation}
{R_N} = {c_3}^N{R_0} + {c_1}\sum\limits_{i = 0}^{N - 1} {{c_3}^i}  = {c_3}^N\left( {{R_0} - \frac{{{c_1}}}{{1 - {c_3}}}} \right) + \frac{{{c_1}}}{{1 - {c_3}}}
\label{e943}
\end{equation}
Rearranging terms results in,
\begin{equation}
\frac{{{R_N} - \frac{{{c_1}}}{{1 - {c_3}}}}}{{{R_0} - \frac{{{c_1}}}{{1 - {c_3}}}}} = {c_3}^N
\label{e944}
\end{equation}
Applying the natural logarithm function to both sides results leads to,
\begin{equation}
\ln \left( {\frac{{{R_N} - \frac{{{c_1}}}{{1 - {c_3}}}}}{{{R_0} - \frac{{{c_1}}}{{1 - {c_3}}}}}} \right) = N\ln {c_3}
\label{e945}
\end{equation}
And the general form for the number of iterations it takes the sweeper swarm to reduce the evader region to be bounded by a circle of a radius that corresponds to the last sweep before completely cleaning the evader region is given by,
\begin{equation}
N = \left\lceil {\frac{{\ln \left( {\frac{{{R_N} - \frac{{{c_1}}}{{1 - {c_3}}}}}{{{R_0} - \frac{{{c_1}}}{{1 - {c_3}}}}}} \right)}}{{\ln {c_3}}}} \right\rceil
\label{e946}
\end{equation}
\section{}
The number of sweep iterations that are required to reduce the evader region to be bounded by a circle with a radius that is less or equal to ${R_{N-2}}$ is calculated in the following manner. We have that,
\begin{equation}
{R_{i + 1}} = {c_3}{R_i} + {c_1}
\label{e915}
\end{equation}
Therefore,
\begin{equation}
{R_{N - 2}} = {c_3}^{N - 2}{R_0} + {c_1}\sum\limits_{i = 0}^{N - 3} {{c_3}^i}  = {c_3}^{N - 2}\left( {{R_0} - \frac{{{c_1}}}{{1 - {c_3}}}} \right) + \frac{{{c_1}}}{{1 - {c_3}}}
\label{e916}
\end{equation}
Rearranging terms results in,
\begin{equation}
\frac{{{R_{N - 2}} - \frac{{{c_1}}}{{1 - {c_3}}}}}{{{R_0} - \frac{{{c_1}}}{{1 - {c_3}}}}} = {c_3}^{N - 2}
\label{e917}
\end{equation}
Therefore the number of sweep iterations that are required to reduce the evader region to be bounded by a circle with a radius that is less or equal to ${R_{N-2}}$ is given by,
\begin{equation}
{R_{N - 2}} = \frac{{{c_1}}}{{1 - {c_3}}} + {c_3}^{N - 2}\left( {{R_0} - \frac{{{c_1}}}{{1 - {c_3}}}} \right)
\label{e918}
\end{equation}
\section{}
The time it takes a multi-agent swarm that performs either the circular sweep or the spiral sweep to completely clean the evader region is calculated as follows. The recursive relation between the next and current radius of the circle that bounds the evader region is given by,
\begin{equation}
{R_{i + 1}} = {c_3}{R_i} + {c_1}
\label{e931}
\end{equation}
Suppose that there exists a constant $\gamma$ such that
\begin{equation}
\gamma {R_i} = {T_i}
\label{e932}
\end{equation}
Therefore multiplying (\ref{e931}) by $\gamma$ on both sides of the equation yields,
\begin{equation}
{T_{i + 1}} = {c_3}{T_i} + {c_4}
\label{e933}
\end{equation}
Where $c_4$ is given by,
\begin{equation}
{{c_4} = \gamma {c_1}}
\label{e934}
\end{equation}
The time it takes to complete the first cycle around the evader region is ${T_0} = \gamma R_0$ and the time it takes to complete the last cycle before the evader region is bounded by a circle of radius $\frac{2r}{n}$ is, that is the time when the evader region is bounded by a circle of with a greater radius than ${R_{N - 1}}$ is given by ${T_{N - 1}} = \gamma {R_{N - 1}}$. Summing over the times of all cycles except the initial one is calculated by summing the cycle times given in (\ref{e933}). This results in,
\begin{equation}
\sum\limits_{i = 1}^{N - 1} {{T_i}}  = {c_3}\sum\limits_{i = 1}^{N - 1} {{T_i}}  + {c_3}({T_0} - {T_{N - 1}}) + \left( {N - 1} \right){c_4}
\label{e58}
\end{equation}
Rearranging terms results in,
\begin{equation}
\sum\limits_{i = 1}^{N - 1} {{T_i}}  = \frac{{{c_3}({T_0} - {T_{N - 1}}) + \left( {N - 1} \right){c_4}}}{{1 - {c_3}}}
\label{e59}
\end{equation}
Since the total time it takes the sweeper swarm to clean the evader region includes also the time of the first sweep we need to add $T_0$ to the summation as well. Thus the total time it takes it takes the sweeper swarm to reduce the evader region to be bounded by a circle of radius that less than or equal to $\frac{2r}{n}$ is given by,
\begin{equation}
T = \sum\limits_{i = 0}^{N - 1} {{T_i}} = \frac{{{T_0} - {c_3}{T_{N - 1}} + \left( {N - 1} \right){c_4}}}{{1 - {c_3}}}
\label{e60}
\end{equation}
\section{}
The time is takes to complete a sweep around the evader region that is bounded by a circle with a radius of ${R_{N-1}}$ is calculated as follows. From Appendix C we have that the recursive relation between the time it takes the sweeping agent to complete sweep number $i$ and the time it takes it to complete sweep number $i+1$ is given by,
\begin{equation}
{T_{i + 1}} = {c_3}{T_i} + {c_4}
\label{e919}
\end{equation}
Therefore,
\begin{equation}
{T_{N - 1}} = {c_3}^{N - 1}{T_0} + {c_4}\sum\limits_{i = 0}^{N - 2} {{c_3}^i}  = {c_3}^{N - 1}\left( {{T_0} - \frac{{{c_4}}}{{1 - {c_3}}}} \right) + \frac{{{c_4}}}{{1 - {c_3}}}
\label{e920}
\end{equation}
Rearranging terms results in,
\begin{equation}
\frac{{{T_{N - 1}} - \frac{{{c_4}}}{{1 - {c_3}}}}}{{{T_0} - \frac{{{c_4}}}{{1 - {c_3}}}}} = {c_3}^{N - 1}
\label{e921}
\end{equation}
Therefore the time is takes to complete a sweep around the evader region that is bounded by a circle with a radius of ${R_{N-1}}$ is given by,
\begin{equation}
{T_{N - 1}} = \frac{{{c_4}}}{{1 - {c_3}}} + {c_3}^{N - 1}\left( {{T_0} - \frac{{{c_4}}}{{1 - {c_3}}}} \right)
\label{e922}
\end{equation}
\section{}
The recursive relation between the next and current radius of the circle that bounds the evader region is given by,
\begin{equation}
{R_{i + 1}} = {c_3}{R_i} + {c_1}
\label{e923}
\end{equation}
Summing over the evader region radii up to the $N_n - 2$ cycle except the initial radius of the evader region is calculated by summing the radii given in (\ref{e923}). This results in,
\begin{equation}
\sum\limits_{i = 1}^{N_n - 2} {{R_i}}  = {c_3}\sum\limits_{i = 1}^{N_n - 2} {{R_i}}  + {c_3}({R_0} - {R_{N_n - 2}}) + \left( {N_n - 2} \right){c_1}
\label{e924}
\end{equation}
Rearranging terms results in,
\begin{equation}
\sum\limits_{i = 1}^{N_n - 2} {{R_i}}  = \frac{{{c_3}({R_0} - {R_{N_n - 2}}) + \left( {N_n - 2} \right){c_1}}}{{1 - {c_3}}}
\label{e925}
\end{equation}
Since the sum of radii in (\ref{e925}) does include the initial radius of the evader region we need to add $R_0$ to the summation as well. Thus the desired sum of radii is given by,
\begin{equation}
\sum\limits_{i = 0}^{{N_n} - 2} {{R_i}}  = \frac{{{R_0} - {c_3}{R_{{N_n} - 2}} + ({N_n} - 2){c_1}}}{{1 - {c_3}}}
\label{e926}
\end{equation}
\section{}
In this appendix the time of inward advancements until the evader region is bounded by a circle with a radius that is smaller or equal to $r$ is computed. This time is denoted by $\widetilde{{T_{in}}}(n) = \sum\limits_{i = 0}^{{N_n} - 2} {{T_{i{n_i}}}}$. This proof continues the derivation from section $\text{III}$. After rearranging terms (\ref{e1100}) resolves to,
\begin{equation}
\widetilde{{T_{in}}}= \sum\limits_{i = 0}^{{N_n} - 2} {{T_{i{n_i}}}}  = \frac{{\left( {{N_n} - 1} \right)r}}{{{V_s} + {V_T}}} - \frac{{2\pi {V_T}\sum\limits_{i = 0}^{{N_n} - 2} {{R_i}} }}{{n{V_s}\left( {{V_s} + {V_T}} \right)}}
\label{e101}
\end{equation}
The term $\sum\limits_{i = 0}^{{N_n} - 2} {\widetilde{{R_i}}}$ is calculated in Appendix $E$. It is given by,
\begin{equation}
\sum\limits_{i = 0}^{{N_n} - 2} {\widetilde{{R_i}}}  = \frac{{\widetilde{{R_0}} - {c_3}\widetilde{{R_{{N_n} - 2}}} + ({N_n} - 2){c_1}}}{{1 - {c_3}}}
\label{e102}
\end{equation}
Where the term $\widetilde{{R_{{N_n} - 2}}}$ is calculated in Appendix $B$. It is given by,
\begin{equation}
\widetilde{{R_{{N_n} - 2}}} = \frac{{{c_1}}}{{1 - {c_3}}} + {c_3}^{{N_n} - 2}\left( {\widetilde{{R_0}} - \frac{{{c_1}}}{{1 - {c_3}}}} \right)
\label{e103}
\end{equation}
Substituting the coefficients in (\ref{e103}) yields,
\begin{equation}
{R_{N - 2}} = \frac{{nr{V_s}}}{{2\pi {V_T}}} + {\left( {1 + \frac{{2\pi {V_T}}}{{n\left( {{V_s} + {V_T}} \right)}}} \right)^{{N_n} - 2}}\left( {\frac{{2\pi {R_0}{V_T} - nr{V_s}}}{{2\pi {V_T}}}} \right)
\label{e104}
\end{equation}
Substituting the coefficients in (\ref{e102}) yields,
\begin{equation}
\begin{array}{l}
\sum\limits_{i = 0}^{{N_n} - 2} {{R_i}}  =  - \frac{{{R_0}n\left( {{V_s} + {V_T}} \right)}}{{2\pi {V_T}}} + \frac{{{n^2}r{V_s}\left( {{V_s} + {V_T}} \right)}}{{{{\left( {2\pi {V_T}} \right)}^2}}}\left( {1 + \frac{{2\pi {V_T}}}{{n\left( {{V_s} + {V_T}} \right)}}} \right)
 + {\left( {1 + \frac{{2\pi {V_T}}}{{n\left( {{V_s} + {V_T}} \right)}}} \right)^{{N_n} - 1}}\left( {\frac{{n\left( {{V_s} + {V_T}} \right)\left( {2\pi {R_0}{V_T} - nr{V_s}} \right)}}{{{{\left( {2\pi {V_T}} \right)}^2}}}} \right)\\
 + \frac{{rn{V_s}({N_n} - 2)}}{{2\pi {V_T}}}
\end{array}
\label{e105}
\end{equation}
Plugging the expression for ${\sum\limits_{i = 0}^{{N_n} - 2} {{R_i}} }$ from (\ref{e105}) into (\ref{e101}) results in,
\begin{equation}
\widetilde{{T_{in}}} = \sum\limits_{i = 0}^{{N_n} - 2} {{T_{i{n_i}}}}  = \frac{{{R_0}}}{{{V_s}}} - \frac{{nr}}{{2\pi {V_T}}}
 - {\left( {1 + \frac{{2\pi {V_T}}}{{n\left( {{V_s} + {V_T}} \right)}}} \right)^{{N_n} - 1}}\left( {\frac{{2\pi {R_0}{V_T} - nr{V_s}}}{{2\pi {V_T}{V_s}}}} \right)
\label{e620}
\end{equation}
\section{}
In this appendix the time of inward advancements until the evader region is reduced to a circle with a radius that is smaller or equal to $2r$ is computed for a swarm that executes the spiral sweep process. This time is denoted by $\widetilde{{T_{in}}}(n) = \sum\limits_{i = 0}^{{N_n} - 2} {{T_{i{n_i}}}}$. This proof continues the derivation from section $\text{IV}$. After rearranging terms (\ref{e702}) resolves to,
\begin{equation}
\sum\limits_{i = 0}^{{N_n} - 2} {{T_{i{n_i}}}}  = \frac{{2r\left( {{N_n} - 1} \right) - \left( {{e^{\frac{{2\pi {V_T}}}{{n\sqrt {{V_s}^2 - {V_T}^2} }}}} - 1} \right)\sum\limits_{i = 0}^{{N_n} - 2} {\widetilde{{R_i}}} }}{{{V_s} + {V_T}}}
\label{e1087}
\end{equation}
The term $\sum\limits_{i = 0}^{{N_n} - 2} {\widetilde{{R_i}}}$ is calculated in Appendix $E$. It is given by,
\begin{equation}
\sum\limits_{i = 0}^{{N_n} - 2} {\widetilde{{R_i}}}  = \frac{{\widetilde{{R_0}} - {c_3}\widetilde{{R_{{N_n} - 2}}} + ({N_n} - 2){c_1}}}{{1 - {c_3}}}
\label{e726}
\end{equation}
Where the term $\widetilde{{R_{{N_n} - 2}}}$ is calculated in Appendix $B$. It is given by,
\begin{equation}
\widetilde{{R_{{N_n} - 2}}} = \frac{{{c_1}}}{{1 - {c_3}}} + {c_3}^{{N_n} - 2}\left( {\widetilde{{R_0}} - \frac{{{c_1}}}{{1 - {c_3}}}} \right)
\label{e725}
\end{equation}
Substituting the coefficients in (\ref{e725}) yields,
\begin{equation}
\widetilde{{R_{{N_n} - 2}}} = - \frac{{2r}}{{1 - {e^{\frac{{2\pi {V_T}}}{{n\sqrt {{V_s}^2 - {V_T}^2} }}}}}}
 + {c_3}^{{N_n} - 2}\left( {\frac{{{R_0}\left( {1 - {e^{\frac{{2\pi {V_T}}}{{n\sqrt {{V_s}^2 - {V_T}^2} }}}}} \right) + r\left( {1 + {e^{\frac{{2\pi {V_T}}}{{n\sqrt {{V_s}^2 - {V_T}^2} }}}}} \right)}}{{1 - {e^{\frac{{2\pi {V_T}}}{{n\sqrt {{V_s}^2 - {V_T}^2} }}}}}}} \right)
\label{e503}
\end{equation}
Denoting by $c_5$ the expression,
\begin{equation}
{c_5} = \frac{{{V_s} + {V_T}}}{{{V_s}\left( {1 - {e^{\frac{{2\pi {V_T}}}{{n\sqrt {{V_s}^2 - {V_T}^2} }}}}} \right)}}{\left( {\frac{{{V_T} + {V_s}{e^{\frac{{2\pi {V_T}}}{{n\sqrt {{V_s}^2 - {V_T}^2} }}}}}}{{{V_s} + {V_T}}}} \right)^{{N_n} - 1}}
\label{e7}
\end{equation}

Substituting the coefficients in (\ref{e726}) yields,
\begin{equation}
\begin{array}{l}
\sum\limits_{i = 0}^{{N_n} - 2} \widetilde{{{R_i}}}  = \frac{{\left( {{R_0} - r} \right)\left( {{V_s} + {V_T}} \right)}}{{{V_s}\left( {1 - {e^{\frac{{2\pi {V_T}}}{{n\sqrt {{V_s}^2 - {V_T}^2} }}}}} \right)}} + \frac{{2r\left( {{V_T} + {V_s}{e^{\frac{{2\pi {V_T}}}{{n\sqrt {{V_s}^2 - {V_T}^2} }}}}} \right)}}{{{V_s}{{\left( {1 - {e^{\frac{{2\pi {V_T}}}{{n\sqrt {{V_s}^2 - {V_T}^2} }}}}} \right)}^2}}}
 - {c_5}\left( {\frac{{{R_0}\left( {1 - {e^{\frac{{2\pi {V_T}}}{{n\sqrt {{V_s}^2 - {V_T}^2} }}}}} \right) + r\left( {1 + {e^{\frac{{2\pi {V_T}}}{{n\sqrt {{V_s}^2 - {V_T}^2} }}}}} \right)}}{{1 - {e^{\frac{{2\pi {V_T}}}{{n\sqrt {{V_s}^2 - {V_T}^2} }}}}}}} \right)\\
 - \frac{{2r({N_n} - 2)}}{{1 - {e^{\frac{{2\pi {V_T}}}{{n\sqrt {{V_s}^2 - {V_T}^2} }}}}}}
\end{array}
\label{e504}
\end{equation}
Plugging the expression for ${\sum\limits_{i = 0}^{{N_n} - 2} {\widetilde{{R_i}}} }$ from (\ref{e504}) into (\ref{e1087}) results in,
\begin{equation}
\begin{array}{l}
\widetilde{T_{in}}(n) = \frac{{2r}}{{{V_s} + {V_T}}} + \frac{{{R_0} - r}}{{{V_s}}} + \frac{{2r\left( {{V_T} + {V_s}{e^{\frac{{2\pi {V_T}}}{{n\sqrt {{V_s}^2 - {V_T}^2} }}}}} \right)}}{{{V_s}\left( {{V_s} + {V_T}} \right)\left( {1 - {e^{\frac{{2\pi {V_T}}}{{n\sqrt {{V_s}^2 - {V_T}^2} }}}}} \right)}}\\
 - \frac{{{{\left( {{V_T} + {V_s}{e^{\frac{{2\pi {V_T}}}{{n\sqrt {{V_s}^2 - {V_T}^2} }}}}} \right)}^{{N_n} - 1}}}}{{{V_s}\left( {{V_s} + {V_T}} \right)\left( {1 - {e^{\frac{{2\pi {V_T}}}{{n\sqrt {{V_s}^2 - {V_T}^2} }}}}} \right)}}\left( {{R_0}\left( {1 - {e^{\frac{{2\pi {V_T}}}{{n\sqrt {{V_s}^2 - {V_T}^2} }}}}} \right) + r\left( {1 + {e^{\frac{{2\pi {V_T}}}{{n\sqrt {{V_s}^2 - {V_T}^2} }}}}} \right)} \right)
\end{array}
\label{e617}
\end{equation}
\section{}
In this appendix we prove that after a spiral sweep the evader region is again circularly shaped. This results hold after every spiral sweep is completed. We will prove it for a two sweeper swarm. The extension to any number of even sweepers is straightforward.
We denote by ${t_\theta }$ the time it takes a sweeper to sweep an angle of $\theta({t_\theta })$ around the evader region. ${t_\theta }$ is given by,
\begin{equation}
{t_\theta } = \frac{{\left( {{R_i} - r} \right)\left( {{e^{\frac{{\theta {V_T}}}{{\sqrt {{V_s}^2 - {V_T}^2} }}}} - 1} \right)}}{{{V_T}}}
\label{e76}
\end{equation}
Similarly we denote by ${t_\pi }$ the time it takes a sweeper to sweep an angle of $\theta({t_\pi })$ around the evader region. ${t_\pi }$ is given by,
\begin{equation}
{t_\pi } = \frac{{\left( {{R_i} - r} \right)\left( {{e^{\frac{{\pi {V_T}}}{{\sqrt {{V_s}^2 - {V_T}^2} }}}} - 1} \right)}}{{{V_T}}}
\label{e77}
\end{equation}
The trajectory of the tip of the sensor that is closest to the center of the evader region of a sweeper that travels counter clockwise is given by,
\begin{equation}
{L_{CCW}}\left( {{t_\theta }} \right) = \left( {{R_i} - 2r + {V_T}{t_\theta }} \right)\left[ {\sin \left( {\theta \left( {{t_\theta }} \right)} \right),\cos \left( {\theta \left( {{t_\theta }} \right)} \right)} \right]
\label{e72}
\end{equation}
The trajectory of the tip of the sensor that is closest to the center of the evader region of a sweeper that travels clockwise is given by,
\begin{equation}
{L_{CW}}\left( {{t_\theta }} \right) = \left( {{R_i} - 2r + {V_T}{t_\theta }} \right)\left[ {\sin \left( {2\pi  - \theta \left( {{t_\theta }} \right)} \right),\cos \left( {2\pi  - \theta \left( {{t_\theta }} \right)} \right)} \right]
\label{e73}
\end{equation}
At time ${t_\theta }$ the counter clockwise sweeper detects the evaders up to the point ${L_{CCW}}\left( {{t_\theta }} \right)$ given in (\ref{e72}). From ${t_\theta }$ to ${t_{\pi }}$ the evaders expand from   ${L_{CCW}}\left( {{t_\theta }} \right)$ in all directions with a maximal speed of $V_T$ for ${t_{\pi }}-{t_\theta }$ time, resulting in a spread of radius ${V_T}\left( {{t_{\pi }} - {t_\theta }} \right)$ in all directions. Hence the wavefront from  ${L_{CCW}}\left( {{t_\theta }} \right)$ is defined by the curve,
\begin{equation}
E\left( {\theta ,\psi } \right) = {L_{CCW}}\left( {{t_\theta }} \right) + {V_T}\left( {{t_\pi } - {t_\theta }} \right)\left[ {\sin \left( \psi  \right),\cos \left( \psi  \right)} \right]
\label{e74}
\end{equation}
for all $\psi  \in \left[ {0,2\pi } \right]$. This shows the expansion from ${L_{CCW}}\left( {{t_\theta }} \right)$ at time ${t_{\pi }}$.
The evaders that spread to the furthest distance from the center of the evader region are the evaders that move along the ray between the center of the region and the position in which the lower tip of the sensor sweeps at time ${t_\theta }$. Combining this insight with the equation in (\ref{e74}) yields that at time ${t_\pi }$ these points satisfy that their distance from the center of the evader region is,
\begin{equation}
R\left( {{t_\pi }} \right) = {R_i} - 2r + {V_T}{t_\theta } + {V_T}\left( {{t_\pi } - {t_\theta }} \right) = {R_i} - 2r + {V_T}{t_\pi }
\label{e75}
\end{equation}
A calculation for the clockwise sweeping agent will result in the exact same distance of the furthest points from the center of the evader region that originated from the right half plane spread of evaders.
Therefore after the completion of the sweep the evader region will be a circle of radius ${R_i} - 2r + {V_T}{t_\pi }$.
\bibliographystyle{unsrt}
\bibliography{Search_for_Smart_Evaders_with_Swarms_of_Sweeping_Agents}
\end{document}